\documentclass[12pt]{article}
\usepackage{eurosym}
\usepackage{natbib}
\usepackage{setspace}
\usepackage{amsmath}
\usepackage{xcolor}
\usepackage{graphicx}
\usepackage{float}
\usepackage{hyperref}
\usepackage[margin=1in]{geometry}
\usepackage{amssymb}
\usepackage{subcaption}
\usepackage[labelfont=bf]{caption}

\newtheorem{proposition}{Proposition}

\newtheorem{definition}{Definition}
\newtheorem{lemma}{Lemma}
\newtheorem{corollary}{Corollary}

\newenvironment{proof}{\paragraph{Proof:}}{\hfill$\square$}

\let\emptyset\varnothing
\hypersetup{colorlinks = true, linkcolor=blue,urlcolor = blue,citecolor=blue}

\begin{document}

\title{A Theory of Investors' Disclosure\thanks{Accepted by Naomi Rothenberg. A previous version of this paper was entitled ``Short, Disclose, and Distort". We would like to thank two anonymous reviewers for their constructive and thoughtful suggestions. We also appreciate helpful
		comments from Janja Brendel, Eti Einhorn, Ilan Guttman, Peicong (Keri) Hu, Xue Jia, Jing Li, Suil Pae, Evgeny Petrov, Chao Tang, Tsahi Versano, Wenfeng Wang, Zheng Wang, Forester Wong, Joanna Wu, Wuyang Zhao, and workshop participants at the AES Asia-Pacific Webinar, Chinese University of Hong Kong, and Tel Aviv University. The work described in this paper was partially supported by grants from the Research Grants Council of the Hong Kong Special Administrative Region,
		China [Project No. CityU 21600421 and HKU 17504622].}}
	\date{}
\author{Pingyang Gao\thanks{University of Hong Kong. Email: pgao@hku.hk} and Jinzhi Lu\thanks{
		City University of Hong Kong. Email: jinzhilu@cityu.edu.hk}}
\maketitle

\begin{abstract}
We investigate investors' voluntary disclosure decisions under uncertainty about their information endowment (\citealt{dye1985disclosure}). In our model, an investor may receive initial evidence about a target firm. Conditional on learning the initial evidence, the investor may receive additional evidence that helps interpret the initial evidence. The investor takes a position in the firm's stock, then voluntarily discloses some or all of their findings, and finally closes their position after the disclosure. We present two main findings. First, the investor will always disclose the initial evidence, even though the market is uncertain about whether the investor possesses such evidence. Second, the investor's disclosure strategy of the additional evidence increases stock price volatility: they disclose extreme news and withhold moderate news. Due to the withholding of the additional evidence, misleading disclosure arises as an equilibrium outcome, where the investor's report decreases (increases) price despite their news being good (bad). These results remain robust when considering the target firm’s endogenous response to the investor's report. 
\end{abstract}

\setlength{\baselineskip}{24pt} \setlength{\parindent}{15pt} 
 \newpage

\section{Introduction}
In recent years, there has been a rise in investors voluntarily disclosing their information to the market. Many investors share their research and analysis on crowdsourced financial information platforms such as Yahoo! Finance and Seeking Alpha (\citealt{dyer2021anonymous} and \citealt{farrell2022democratization}). Furthermore, an increasing number of activist short sellers publicly disseminate their research reports on target firms (\citealt{ljungqvist2016constraining} and \citealt{brendel2021responding}). Despite the growing empirical interest in investors' disclosure, there is little theoretical guidance on their disclosure strategies. Prior literature on voluntary disclosure has focused almost exclusively on the disclosure strategies of corporate managers (e.g., \citealt{dye1985disclosure} and \citealt{verrecchia1983discretionary}). In this paper, we study investors' communication to the capital market by highlighting one major difference of the disclosure by managers and investors. All managerial incentive systems are designed to better align the interests of managers with those of the shareholders. As a result, managers often have long exposure to the firm's stock prices and thus aim to maximize stock prices (\citealt{beyer2010financial}). In contrast, investors can flexibly choose their positions on target firms before making the disclosure. Therefore, we examine investors’ disclosure strategies by augmenting \cite{dye1985disclosure} with a key departure that the sender (investors) can choose a position before disclosure. 

In our model, an investor learns an initial signal about a target firm with some probability. Conditional on learning the initial signal, the investor learns an additional signal with some probability. The investor first takes a position in the firm's stock, then voluntarily discloses some or all of their information, and finally closes their position after the disclosure. To illustrate the idea, suppose an investor has identified a red flag at a coffee chain: a sample of customer receipts shows that items per order has recently dropped. This could be negative news for the company if the demand of existing customers is declining. However, the news could be positive if the company is attracting new customers who may place smaller orders. Conversely, an increase in items per order, a ``green flag", could be positive news if existing customers are ordering more, but could be negative news if the company is losing less loyal customers. The investor may uncover additional evidence that helps better interpret the value implications of the initial information. Thus, the investor's two signals are complements, as in \cite{cheynel2020publicMosaic}.

Our first finding establishes that the investor informed about the initial evidence always discloses it despite the market's uncertainty about whether they possess such evidence. This result contrasts with prior work such as \cite{dye1985disclosure}. In \cite{dye1985disclosure}, the manager with long exposure to the stock price withholds their news if its revelation lowers the stock price. In our model, the investor can trade in both directions before disclosing and thus profit from price movements in either direction. In other words, the investor's objective is to maximize \textit{price volatility} through their disclosure. In contrast, if the investor withholds the initial evidence, the price will remain unchanged, and their profit will be zero. Therefore, the investor will never remain silent.

Our second result shows that the investor discloses the additional evidence if and only if it induces sufficiently large price volatility: they disclose extreme news and withhold moderate news. In the previous example, first consider the case where the initial signal is a red flag and is bad news on average. Thus, the disclosure of the initial signal alone generates a negative price impact. If the additional evidence strongly indicates that the coffee company is facing a reduction in customer demand, its disclosure reinforces the initial negative price impact. But if the additional evidence shows convincingly that the company is attracting new customers, then its disclosure reverses the initial negative price impact and leads to a sufficiently positive price impact.\footnote{This is supported by anecdotal evidence. For example, in \href{https://www.muddywatersresearch.com/research/bol/complexity-creating-arbitrage/}{Muddy Waters Research's report on Bollore}, the famous short-seller first points out a red flag on Bollore, a ``horrifically complex corporate structure". However, in the same report, the short-seller provided detailed evidence and reasons for why the opacity of Bollore is actually the reason to take a long position in the company.} In both cases, the disclosure of the additional evidence increases price volatility compared with the disclosure of the initial signal alone. However, when the additional evidence is moderate, its disclosure only offsets the initial negative price impact, resulting in a marginal change in price volatility. Therefore, the investor will withhold the additional evidence. 

Similar intuitions apply when the investor's initial signal is a green flag. If the additional evidence shows convincingly that the news is very good (e.g., existing customers are placing bigger orders) or very bad (e.g., the company is losing less loyal customers), the investor will disclose the additional evidence along with the green flag. When the additional evidence is moderate, the investor will withhold it.

We show that the withholding of the additional evidence results in misleading disclosure. Specifically, when the disclosure of additional evidence reverses the impact of the initial evidence by only a small magnitude, the investor withholds it and issues a report with only the initial evidence. The market is misled by such a report, as the price would move in the opposite direction if the investor disclosed the additional evidence. This finding implies that both activist short-sellers and long investors can issue misleading disclosures. It is consistent with partial and misleading disclosure being one of the common complaints from target firms about activist short-sellers. In fact, \cite{brendel2021responding} documents that 12\% of target firms in their sample respond to activist short sellers by providing additional disclosure.\footnote{For example, in 2011, short-seller Citron Research published \href{https://cdn.gmtresearch.com/public-ckfinder/Short sellers/Citron\%20Research/Citron\%20Research_QIHOO_Citron\%20Reports\%20on\%20QIHOO\%20360_Nov1-2011.pdf}{a short selling report} on the Chinese internet company Qihoo 360. Citron's report pointed out several red flags for Qihoo 360's user number. However, Qihoo's \href{https://www.prnewswire.com/news-releases/qihoo-360-responds-to-citron-research-report-133030163.html}{response} clarifies that Citron's analysis on their user number is incomplete and misleading.}

We then conduct comparative statics to explore the model's empirical predictions and show how the endogenous nature of the investor's disclosure interacts with those predictions. We focus on two parameters, the investor's competence modeled as the probability that the investor receives the additional information, and the firm's information environment modeled as the probability that the firm's fundamental is revealed by sources other than the investor. We show that a more competent investor is less likely to withhold information and engage in misleading disclosure, but the price reaction to their publication of a simple report (the initial signal alone without the additional evidence) is weaker after controlling for the report's content. The last prediction, which may disagree with one's intuition that market response to a report increases with the investor's competence, is driven by the investor's endogenous disclosure. Since a more competent investor is more likely to discover additional evidence, the market, upon observing a simple report from such an investor, ascribes a larger probability that they have withheld the additional evidence and responds with a lower price reaction. 

When the firm's information environment improves, the price reacts more strongly to the investor’s simple report. Moreover, the investor becomes less likely to engage in misleading disclosure and withholds more additional evidence that aligns with the direction of the initial evidence. The intuition is as follows: Information from other sources confronts the investor's misleading disclosure, resulting in a loss to their position. Therefore, the firm’s information environment plays a disciplinary role on misleading disclosure.  As misleading disclosure decreases, the market reacts more strongly to a simple report, as such a report becomes more convincing on average.  Anticipating this change in market beliefs, the investor withholds more additional evidence that aligns with the direction of the initial evidence, to benefit from the higher price volatility induced by a simple report. 

Finally, we extend the baseline model to consider the target firm’s endogenous response to the investor's report. Our findings show that full disclosure of the initial evidence remains an equilibrium. Furthermore, consistent with the baseline model, the investor discloses extreme additional evidence and withholds moderate additional evidence.

\subsection{Related Literature}

We contribute to the vast literature on verifiable disclosure.\footnote{See, for example, \cite{verrecchia1983discretionary}, \cite{dye1985disclosure}, \cite{jung1988disclosure}, and \cite{shin1994news}. Readers are referred to a number of excellent reviews of this vast literature, including \cite{verrecchia2001JAEessays}, \cite{dye01jae}, \cite{beyer2010financial}, \cite{StockenSurveyDisclsoure} and \cite{bertomeuCheynelSuvery2015COC}.} We extend \cite{dye1985disclosure} by incorporating
two salient features: investors' flexible choice of initial position and multi-dimensional information. These departures from \cite{dye1985disclosure}
drive our main results, namely that the investor always discloses the
initial evidence and that the disclosure of additional information is two-tailed (i.e., the investor discloses extreme news and withholds moderate news). A similar two-tailed disclosure strategy also arises in  \cite{beyer2012voluntary} and \cite{kumar2012voluntary}, but due to different mechanisms.\footnote{\cite{beyer2012voluntary} introduce voluntary disclosure into the model of \cite{myers1984corporate}. In their model, the manager's costly signaling
	via over-investment drives the endogenous cost of disclosure, and this cost
	is highest for intermediate types. As a result, the types in the middle
	strategically withhold their information. In \cite{kumar2012voluntary}, the firm’s manager cares about both the price and the value of the firm. They show that the disclosure of good news increases both short-term prices and investment efficiency. Furthermore, the manager discloses extremely bad news because investment efficiency gains from accurate capital allocation outweigh the short-term negative price impact. For intermediate news, managers remain silent as the investment distortions are less severe and do not justify the adverse price reaction.} 
	
Our paper is related to \cite {cheynel2020publicMosaic}, which introduces an information mosaic to model multi-dimensional information where the fundamental is the multiplication of
two components. We have borrowed their information structure to study a
different problem. In our model, the investor has information about both
components and decides how to disclose them. Also related are the
studies where the manager's reporting objective (exposure to the stock
price) is uncertain but exogenous (\citealt{einhorn07jae} and \citealt{fischer2000reporting}). In contrast, the investor endogenously chooses their exposure in our model. 

We have adopted the verifiable disclosure approach to study investors' disclosure. Corporate communication in practice ranges from easily verifiable content to potentially fraudulent reports to sometimes even completely unverifiable information. Accordingly, the corporate disclosure literature has used different approaches ranging from verifiable disclosure (e.g., \citealt{dye1985disclosure}) to costly misreporting (e.g., \citealt{gao2013conservatism}) to cheap talks (e.g., \citealt{stocken2000credibility}). Similarly, investors' communications to capital markets are bound to be diverse as well. At least some communication by investors to the capital market fits the verifiable disclosure framework.\footnote{For example, reports by activist short sellers are often very detailed, providing ``smoking gun'' evidence such as audio and video clips (\citealt{ljungqvist2016constraining}). Additionally, investors' reports are frequently scrutinized in practice. Activist short sellers' reports, for instance, can be challenged by target firms and their significant stakeholders (\citealt{brendel2021responding}). Articles on Seeking Alpha undergo thorough editorial review before publication (\citealt{farrell2022democratization}) and are monitored by regulators (\citealt{dyer2021anonymous}).  Finally, \cite{ljungqvist2016constraining} document that other investors, particularly long investors in target companies, respond strongly to activist short sellers’ reports, suggesting these reports are credible.}

Therefore, our model complements the existing theories that have examined investors’ disclosure based on the unverifiable disclosure approach (e.g., \citealt{benabou1992using}, \citealt{van2003rumors}, \citealt{marinovic2013forecasters}, and \citealt{schmidt2020stock}). These models typically feature multiple equilibria, including a babbling equilibrium with no communication.\footnote{The investor's credibility is modeled by assuming an honest type who always tells the truth (\citealt{benabou1992using}) or by using mechanisms such as varying trading horizons (\citealt{schmidt2020stock}).}  In contrast, our model adopts the verifiable disclosure approach, resulting in a unique equilibrium where misleading disclosure can arise as an equilibrium outcome. We believe that both approaches can contribute to our understanding of investors' disclosure.

Our paper also contributes to our understanding of investors' disclosure and how
their information is incorporated into stock prices. Prior literature has
focused on the channel through which investors' information indirectly
impacts stock prices via their trading activities (e.g., \citealt{kyle85e}, 
\citealt{glosten1985bid}, \citealt{goldstein2008manipulation}, and \citealt{gao2025manipulation}). In contrast, our model focuses on the disclosure channel. We show that misleading disclosure can happen even though everyone is rational and no one is systematically fooled. Our results are useful to help guide the growing empirical literature on investors' disclosure (see \citealt{ljungqvist2016constraining}, \citealt{zhao2020activist}, \citealt{paugam2020deploying}, \citealt{brendel2021responding}, \citealt{dyer2021anonymous}, \citealt{farrell2022democratization}, and \citealt{mcdonough2024voluntary}). We discuss the empirical implications in detail in Section \ref{Section: cs}.

Finally, there is a sizeable literature on disclosure with multiple private signals (e.g., \citealt{dyeFinn2007equilibrium}, \citealt{arya2010discretionary}, and \citealt{ebertstecher2017discretionary}). Our study is related to \cite{pae2005selective} and \cite{beyer2023disclosure}. Both studies show that managers may engage in selective and misleading disclosure. In their models, the manager maximizes the price and therefore, withholds signals with low realizations. In contrast, the investor in our model adopts a very different disclosure strategy due to their flexible choice of initial position.\footnote{The feature of multiple signals also manifests in those dynamic settings when the sender receives multiple signals over time (e.g., \citealt{einhorn2008intertemporal}, \citealt{bertomeu2011low}, \citealt{guttman2014not}, \citealt{bertomeu2020often}, \citealt{bertomeu2022dynamics}). And yet another variant takes the form that the additional signal is about the precision of the initial signal (e.g., \citealt{hughes2004voluntary}, \citealt{bagnoli2007financial}) or comes from other sources (e.g., \citealt{einhorn2005nature} and \citealt{einhorn2018competing}).}

The rest of the paper is organized as follows. We present the model in Section \ref{Section: Model}, analyze the equilibrium in Section \ref{Section: Equilibrium}, and develop the empirical implications in Section \ref{Section: cs}. Section \ref{Section: Extension} studies the extension with the firm's endogenous response. Section \ref{Section: Conclusion} concludes. 

\section{Model\label{Section: Model}}

Our model augments \cite{dye1985disclosure} with trading. An investor
may receive two signals about the target firm. The investor trades in the
firm's stock, makes disclosure decisions, and closes the position afterward.

\noindent \textbf{Information Structure:} A firm's value $\theta$ has two
components: 
\begin{equation}
\theta =rx.
\end{equation}
$r$ is a continuous variable on $\mathbb{R^{+}}$, with density function $f(.)$ and mean $r_0$. $x$ is a continuous variable on $\mathbb{R}$, with density function $g(.)$ and mean $\mu_{0}$.  We assume $\mu_0>0$ without loss of generality. 

An investor may receive information about both $r$ and $x$. They learn $r$ with probability $\alpha$ and, conditional on learning $r$, they further learns about $x$ with probability $\beta$. To ease exposition, we say the investor is \textit{uninformed} if they learn no information, \textit{partially informed }if they learn only $r$, and \textit{informed} if they learn both $r$ and $x$. 

We interpret $r$ as the initial information and $x$ as additional evidence that helps interpret the initial evidence. For instance, $r$ could represent a decrease in items per order, while $x$ could represent information about whether this decline is due to decreased demand from existing customers or the company attracting new customers who place smaller orders.

\noindent \textbf{Trading:} The investor trades twice, establishing a
position before the disclosure and closing it after the disclosure. The size
of their order $\rho$ is constrained to be $\rho \in [-1,1]$, where $\rho <0$ stands for a short position, $\rho >0$ stands for a long position, and $0$ stands for no trade. 

\noindent \textbf{The investor's disclosure:} Between the two rounds of trading, the investor makes their disclosure decision. As in \cite{dye1985disclosure} and \cite{jung1988disclosure}, any disclosure is assumed truthful. The informed type can disclose both signals, or imitate the partially informed type by withholding $x$, or imitate the uninformed type by withholding both $r$ and $x$. The partially informed type can disclose $r$, or imitate the uninformed type by withholding $r$. The unformed type cannot disclose anything or prove that they are uninformed. We say the investor's report is simple if they only disclose $r$ and is elaborate if they disclose both $r$ and $x$.\footnote{The investor's position (long or short) is not necessarily disclosed in the report. Our results remain the same regardless of whether the position is disclosed, as the market can rationally infer it from the investor's report.}

\noindent \textbf{Disclosure of $\theta$ by other information sources:}
After the investor's disclosure, we assume that $\theta$ is revealed with probability $q$. This feature is motivated by the idea that other market participants, including media, analysts, other investors, and the target firm itself, may also disclose relevant information that confronts any omission or withholding of information in the investor's disclosure.\footnote{For example, \cite{brendel2021responding} documents that 31\% of the target firms in their sample respond to activist short sellers' reports, and 12\% of them respond by providing additional disclosure. Target firms may not respond for many reasons. For example, it takes time to collect the necessary evidence. Additionally, responses may reveal proprietary information and incur legal liabilities (\citealt{brendel2021responding}).} Thus, we can view $q$ as a proxy for the firm's information environment.

The timeline of the model is as follows:
\begin{itemize}
\item At date 1, the investor learns $r$ with probability $\alpha$. Conditional on learning $r$, the investor learns $x$ with probability $\beta$. After this information event, the investor establishes an initial position $\rho $ at price $P_{1}=E[\theta]=r_0\mu_0$.

\item At date 2, the investor makes the disclosure decisions. After the investor's disclosure, the firm value $\theta$ is revealed with probability $q$.

\item At date 3, the investor closes their initial position at the updated stock
price $P_{2}$.
\end{itemize}

We define misleading disclosure as situations where the stock price moves in the opposite direction of the investor's information.
\begin{definition}[Misleading Disclosure]\label{def:misleading}
	The informed investor's simple report is misleading if it decreases (increases) the price when $rx>r_0\mu_0$ ($rx< r_0\mu _{0}$).
\end{definition}
Note that when the investor engages in misleading disclosure, they also trade against their information, as they take a short (long) position despite having good (bad) news.

The equilibrium concept we adopt is Grossman–Perry–Farrell equilibrium (GPFE), as defined in \cite{bertomeu2018verifiable}. This concept is a refinement of the perfect Bayesian equilibrium (PBE) and is used to determine the market's belief upon potential off-equilibrium disclosures. A detailed definition is provided in the proof of Proposition \ref{THM: r-disclosure_continuous} in the appendix.

\section{Equilibrium \label{Section: Equilibrium}}
\subsection{Preliminary Analysis}
We first pin down the investor's disclosure strategy of the initial evidence $r$.
\begin{proposition}
	\label{THM: r-disclosure_continuous} The informed investor and the partially informed investor always disclose the initial evidence $r$, regardless of its realization.
\end{proposition}

Proposition \ref{THM: r-disclosure_continuous} shows that the investor always discloses the initial evidence even though the market is uncertain about whether they have received it or not. This unraveling result is different from the partial disclosure result found in \cite{dye1985disclosure} and \cite{jung1988disclosure} because the investor in our setting can trade in both directions before the disclosure. The sender (manager) in \cite{dye1985disclosure} has long exposure to the stock price and benefits only from positive price movement. When the news is negative and its disclosure moves the stock price downward, the manager has incentives to withhold such information. In contrast, the trading flexibility of the sender (investor) in our model enables him to profit from any price movement, regardless of its direction. They can make a profit by taking a short (long) position if their report decreases (increases) the price. Thus, the investor tries to move the price as much as possible, regardless of the direction. Anticipating this incentive, if the investor does not disclose any report, the market infers that the investor must be uninformed about $r$, and as a result, the price does not move. Thus, the uninformed investor does not trade, and the investor always disclosing $r$ is an equilibrium. 

The proof of the equilibrium's uniqueness requires that we determine the market's belief upon the off-equilibrium disclosure of $r$. In the appendix, we apply neologism-proofness as in \cite{bertomeu2018verifiable} to accomplish the task. Readers can also refer to Lemma \ref{THM: r-disclosure_single} in the appendix for the unraveling result in the case of a single signal. 

Moreover, when the investor discloses $r$ and withholds $x$, the potential revelation of the firm's value by other information sources will dilute their profit. However, this does not prevent the investor from disclosing $r$, as withholding it would result in zero profit. If the investor expects to incur negative profit from withholding $x$, they can always switch to full disclosure of both $r$ and $x$ and take positions accordingly, thereby guaranteeing positive profit. 

\subsection{Main Results}
Having established that the investor will always disclose $r$, we turn to the investor's disclosure of the additional evidence $x$ and the trading decision $\rho$. From this point forward, we refer to the disclosure of $x$ whenever discussing disclosure, and ``investor" will mean the informed investor unless stated otherwise.

Recall that due to the flexible choice of initial position, the investor can make a profit as long as their disclosure moves the stock price away from the prior ($r_0\mu_0$). In other words, the investor tries to increase \textit{price volatility} through their disclosure. We use $\mu_{ND}$ to denote the market's posterior expectation about the additional information $x$ upon observing only the initial signal $r$. Thus, the price volatility of a simple and an elaborate report, which we denote by $\pi_q(x)$ and $\pi(x)$, are given by:\footnote{Note that at date 3, the firm value is revealed by other information sources with probability $q$. Thus, when withholding $x$ (i.e., issuing a simple report), the investor expects the date 3 price to be $qrx+(1-q)r\mu_{ND}$.}
\begin{align}
	&\pi_q(x)=|qrx+(1-q)r\mu_{ND}-r_0\mu_0|, \label{piq}\\
	&\pi(x)=|rx-r_0\mu_0|. \label{pix}
\end{align}
The next lemma characterizes the investor's objective when making the disclosure decision. 
\begin{lemma}\label{volatility}
The informed investor withholds $x$ if and only if $\pi_q(x)\ge\pi(x)$.
\end{lemma}

Lemma \ref{volatility} shows that when determining the investor's disclosure strategy, we can put aside the investor's trading decisions and focus on price volatility. Given that $\mu_{ND}$ is a constant in equilibrium, it is straightforward that $\pi_q(x)\ge\pi(x)$ if and only if the firm value (i.e., $rx$) is sufficiently close to its prior. This implies that, fixing the initial evidence $r$, the investor withholds $x$ if and only if it is in the middle. In the appendix, we show that the equilibrium exists and is uniquely determined by two disclosure thresholds $\bar{x}$ and $\underline{x}$. Thus, we have our main result below.
\begin{proposition} \label{x disclosure}
For any realization of $r$, the informed investor's unique equilibrium disclosure and trading
decisions are characterized by three regions:
	\begin{enumerate}
		\item	For $x\in (-\infty,\underline{x})$, the investor takes a short position and discloses $x$.
	
		\item   For $x\in (\bar{x},+\infty)$, the investor takes a long position and discloses $x$.
		
		\item	For $x\in [\underline{x},\bar{x}]$, the investor withholds $x$ and takes a short (long) position when $r<r_0$ ($r>r_0$).
	\end{enumerate}
The partially informed investor takes a short (long) position when $r<r_0$ ($r>r_0$). When $r<r_0$, $\bar{x}$ and $\underline{x}$ are uniquely determined by equations \eqref{shortindiff1} and \eqref{shortindiff2}. When $r>r_0$,  $\bar{x}$ and $\underline{x}$ are uniquely determined by equations \eqref{longindiff1} and \eqref{longindiff2}. 
\end{proposition}

Proposition \ref{x disclosure} is illustrated in Figure \ref{equilibrium2} and Figure \ref{equilibrium3}. The investor issues an elaborate report in region A and region D. They take a short position in region A and a long position in Region D. In these two regions, the additional evidence is so significant that the price volatility of an elaborate report exceeds that of a simple report. In contrast, when the additional evidence is weak (i.e., in regions B and C), disclosing it would only offset the price impact of the initial evidence. Thus, the investor pretends to be uninformed about the additional evidence, to profit from the higher price volatility induced by a simple report that contains only the initial evidence. 

In regions B and C, the investor's trading position depends on the realization of the initial evidence. To see the intuition, consider first the market's reaction to a simple report with a red flag (i.e., $r<r_0$). Note that when only the initial evidence is disclosed, the market rationally believes that the investor either failed to uncover any additional evidence or chose not to disclose it. In the former case, the belief about the second component of the firm value $x$ does not change, but the first component of the firm value is confirmed to be lower than the prior due to $r<r_0$. Overall, the firm value is revised downward from $r_0\mu_0$ (i.e., the prior) to $r\mu_{0}$.  In the latter case, the nature of the news depends on the equilibrium disclosure strategy about $x$, but the news cannot be too favorable, as otherwise the investor would have disclosed it after taking a long position. Therefore, despite the possibility of strategic withholding of $x$, the market should still react negatively to reports containing only the red flag, and rationally anticipates that the investor takes a short position when issuing a report with only the red flag. As shown in Figure \ref{equilibrium2}, at the lower disclosure threshold, the investor is indifferent between shorting plus withholding and shorting plus disclosing, so the non-disclosure price satisfies $r\mu_{ND} = r\underline{x}$.

\begin{figure}[H]
	\centering
	\includegraphics[width=1\textwidth]{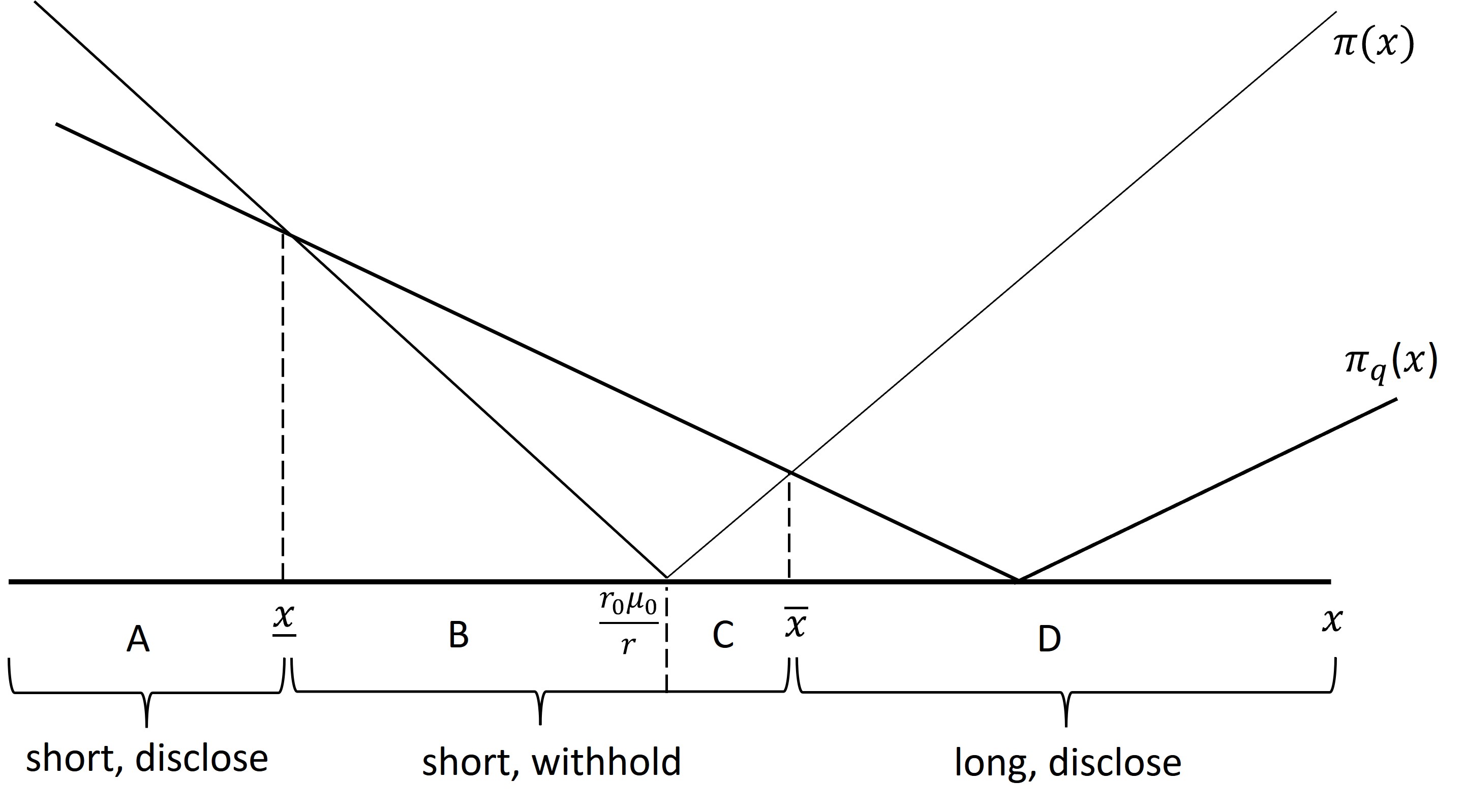}
	\caption{\textbf{The equilibrium when $r<r_0$\\Non-disclosure price $r\mu_{ND}=r\underline{x}$}}
	\label{equilibrium2}
\end{figure}

\begin{figure}[H]
	\centering
	\includegraphics[width=1\textwidth]{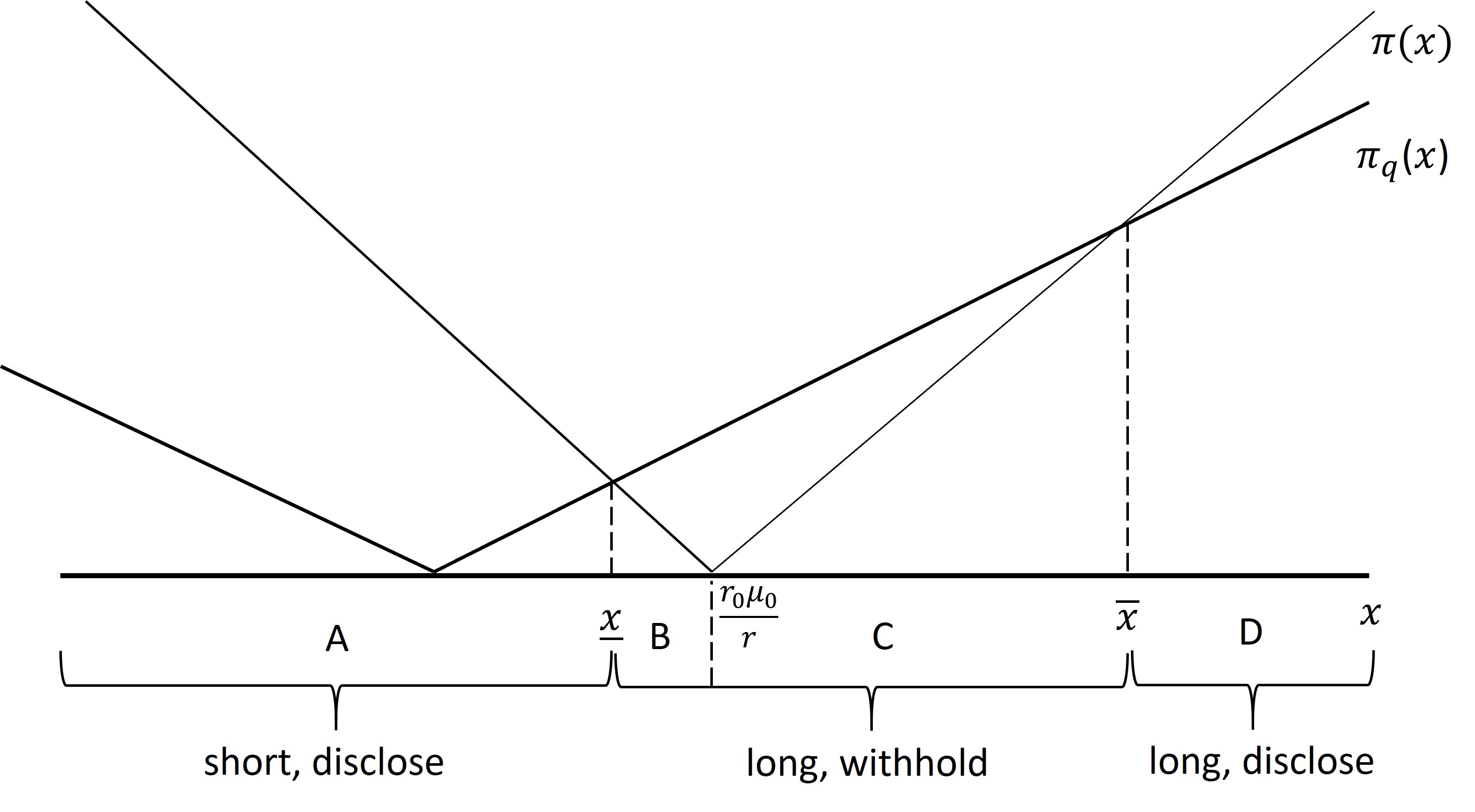}
	\caption{\textbf{The equilibrium when $r>r_0$\\Non-disclosure price  $r\mu_{ND}=r\bar{x}$}}
	\label{equilibrium3}
\end{figure}

Similarly, when the investor issues a simple report with a green flag (i.e., $r>r_0$), the market reacts positively and rationally anticipates that the investor takes a long position when issuing a report with only the green flag. As shown in Figure \ref{equilibrium3}, at the upper disclosure threshold, the investor is indifferent between buying plug withholding and buying plus disclosing. Thus, the non-disclosure price satisfies $r\mu_{ND}=r\bar{x}$.

Building on Proposition \ref{x disclosure}, we next show that misleading disclosure arises in equilibrium.
\begin{corollary} \label{misleading}
Misleading disclosure arises when $x\in[\frac{r_0\mu_0}{r},\bar{x}]$ for $r<r_0$ (region C in Figure \ref{equilibrium2}) and when $x\in [\underline{x},\frac{r_0\mu_0}{r}]$ for $r>r_0$ (region B in Figure \ref{equilibrium3}).   
\end{corollary}

Corollary \ref{misleading} shows that misleading disclosure arises in a sub-region of the withholding region. At point $x=\frac{r_0\mu_0}{r}$ in Figure \ref{equilibrium2}, the negative price impact of the red flag is perfectly offset by the disclosure of additional evidence. Moving from $x=\frac{r_0\mu_0}{r}$ towards the left, the additional evidence indicates bad news, and the price indeed decreases.  Moving from $x=\frac{r_0\mu_0}{r}$ towards the right (region C of Figure \ref{equilibrium2}), the additional evidence indicates good news. However, in this case, the investor reports the red flag while withholding the additional evidence, leading to a price decline. This behavior aligns with the concept of misleading disclosure as defined in Definition \ref{def:misleading}. Similarly, in region B of Figure \ref{equilibrium3}, the investor reports the green flag while concealing additional evidence that indicates bad news, which is again consistent with Definition \ref{def:misleading}.

\section{Empirical Implications \label{Section: cs}}
Having established the equilibrium in Proposition \ref{THM: r-disclosure_continuous} and Proposition \ref{x disclosure}, we derive empirical predictions of the model by conducting comparative statics analysis. 

Our predictions focus on the following dependent variables. First, recall that we define the investor's report to be simple if they only disclose $r$ and to be elaborate if they disclose both $r$ and $x$. We define the frequency of negative (positive) elaborate reports as $\Pr(x<\underline{x})$ ($\Pr(x>\bar{x})$). We define the report's overall elaborateness as $\Pr(x<\underline{x})+\Pr(x>\bar{x})$, which is the likelihood of an elaborate report conditional on the investor being informed about $x$.\footnote{An alternative definition of elaborateness is $\beta(\Pr(x<\underline{x})+\Pr(x>\bar{x}))$, which represents the unconditional likelihood of an elaborate report. All our predictions remain qualitatively unchanged under this definition.} Note that elaborateness is also a measure of the informativeness of date 3 stock price: when the investor discloses both signals, the price is equal to the true value of the firm. When the investor issues a simple report, the price is only a noisy measure of firm value. Empirically, investors’ reports are more continuous. The elaborateness of reports could be proxied by the length or the number of points. 

Second, a related variable is the extremity of reports. A report is more extreme in our model if the news is further away from the prior. The investor discloses the news in two tails and withholds in the middle. As a result, the disclosed news is more extreme than the withheld news. The implication is that as the elaborate reports become more frequent, they become less extreme as well. We define the extremity of negative (positive) elaborate reports as $\Pr(x\in[\underline{x},\frac{r_0\mu_0}{r}])$ ($\Pr(x\in[\frac{r_0\mu_0}{r},\bar{x}])$),  and overall extremity as $\Pr(x\in[\underline{x},\bar{x}])$.

Third, note that stock prices endogenously react to the investor's disclosure in equilibrium. Conditional on a simple report and no information revelation by other information sources, the magnitude of market reaction is captured by the price volatility, $\pi_{q=0}(x)=|r_0\mu_0-r\mu_{ND}|$, where $\pi_q$ is defined in equation \eqref{piq}. We have demonstrated in Proposition \ref{x disclosure} that when $r<r_0$ ($r>r_0$), $\pi_{q=0}(x)=r_0\mu_0-r\underline{x}$ ($\pi_{q=0}(x)=r\bar{x}-r_0\mu_0$).

Fourth, our model also generates predictions about misleading disclosure. Misleading disclosure is not directly observable in general, and we leave it to empirical research to find the proper proxy.\footnote{For an indirect proxy of misleading disclosure, one may look for cases where the short run stock price reaction to a report is followed by a reversal in the long run.} As we have discussed earlier in Corollary \ref{misleading}, misleading disclosure is captured by region C in Figure \ref{equilibrium2} and region B in Figure \ref{equilibrium3}. All else equal, misleading disclosure reduces the informativeness of the price, as the price is most informative when both pieces of information are disclosed. 

We start with the predictions on the elaborateness of reports.
\begin{proposition}
	\label{maincs}
	\begin{enumerate}
	  \item The elaborateness of reports increases with the investor's competence. That is, $\frac{\partial \bar{x}}{\partial \beta}<0$ and $\frac{\partial \underline{x}}{\partial\beta}>0$. 	
	  
	  \item The firm’s information environment ($q$) has an ambiguous effect on the elaborateness
	  of reports.
	  	\begin{enumerate}
	  	 \item When $r<r_0$, the frequency of negative (positive) elaborate reports decreases (increases) with the firm's information environment: $\frac{\partial \bar{x}}{\partial q}<0$ and $\frac{\partial \underline{x}}{\partial q}<0$. 
	  	 \item When $r>r_0$, the frequency of positive (negative) elaborate reports decreases (increases) with the firm's information environment: $\frac{\partial \bar{x}}{\partial q}>0$ and $\frac{\partial \underline{x}}{\partial q}>0$. 
	  	\end{enumerate}
	\end{enumerate}
\end{proposition}

Part 1 of Proposition \ref{maincs} carries over to our setting the insight from the prior literature that the sender's disclosure is increasing in their information endowment (e.g., \citealt{jung1988disclosure}). As the investor's competence $\beta$ increases, the disclosure of news in both tails increases, causing the non-disclosure region (the joint size of B and C in Figures \ref{equilibrium2} and \ref{equilibrium3}) to shrink, which maps into an increase in the report's elaborateness. Intuitively, the market will be more skeptical about a simple report issued by a more successful hedge fund (e.g., funds with better track records), as it is more likely that such a fund has chosen to conceal additional evidence, rather than having failed to discover it. Thus, more successful hedge funds are more likely to provide elaborate reports.

To understand Part 2 of Proposition \ref{maincs}, first note that when $q$ increases, an investor's misleading disclosure is more likely to be confronted by other sources, and thus the investor issues less misleading disclosure. In other words, as $q$ increases, region C in Figure \ref{equilibrium2} and region B in Figure \ref{equilibrium3} shrink. This explains why $\frac{\partial \bar{x}}{\partial q}<0$ when $r<r_0$ and why $\frac{\partial \underline{x}}{\partial q}>0$ when $r>r_0$. The disclosure of $\theta$ by other information sources holds the investor accountable and promotes honesty. 

On the other hand, the increase in $q$ indirectly affects the investor's disclosure through the pricing of non-disclosure of $x$, compelling the investor to withhold more information. The intuition is as follows: As $q$ increases, we already know that the investor issues less misleading disclosure. As a result, the market reacts more strongly to a simple report, as such a report becomes more convincing on average.  Anticipating this change in market beliefs, the investor withholds more additional evidence that aligns with the direction of the initial evidence, to benefit from the higher price volatility induced by a simple report. This implies that as $q$ increases, region B in Figure \ref{equilibrium2} and region C in Figure \ref{equilibrium3} expand. Thus, we have $\frac{\partial \underline{x}}{\partial q}<0$ when $r<r_0$ and $\frac{\partial \bar{x}}{\partial q}>0$ when $r>r_0$. 

Since the likelihood of negative elaborate reports and the likelihood of positive elaborate reports move in opposite directions, the overall effect of the firm's information environment on the elaborateness of reports is ambiguous.  We conduct numerical analysis to better understand the effect, assuming $x$ is normally distributed with mean $\mu_0$ and variance $\sigma_x^2$. We find that when $\beta$ is small, the elaborateness of reports increases with $q$, as shown in Figure \ref{q_numerical1}. However, the effect can be non-monotonic when $\beta$ is large (i.e., when the investor has high expertise). One such example is shown in Figure \ref{q_numerical2}. 

\begin{figure}[H]
	\centering
	\begin{subfigure}[b]{0.45\textwidth}
		\centering
		\includegraphics[width=\textwidth]{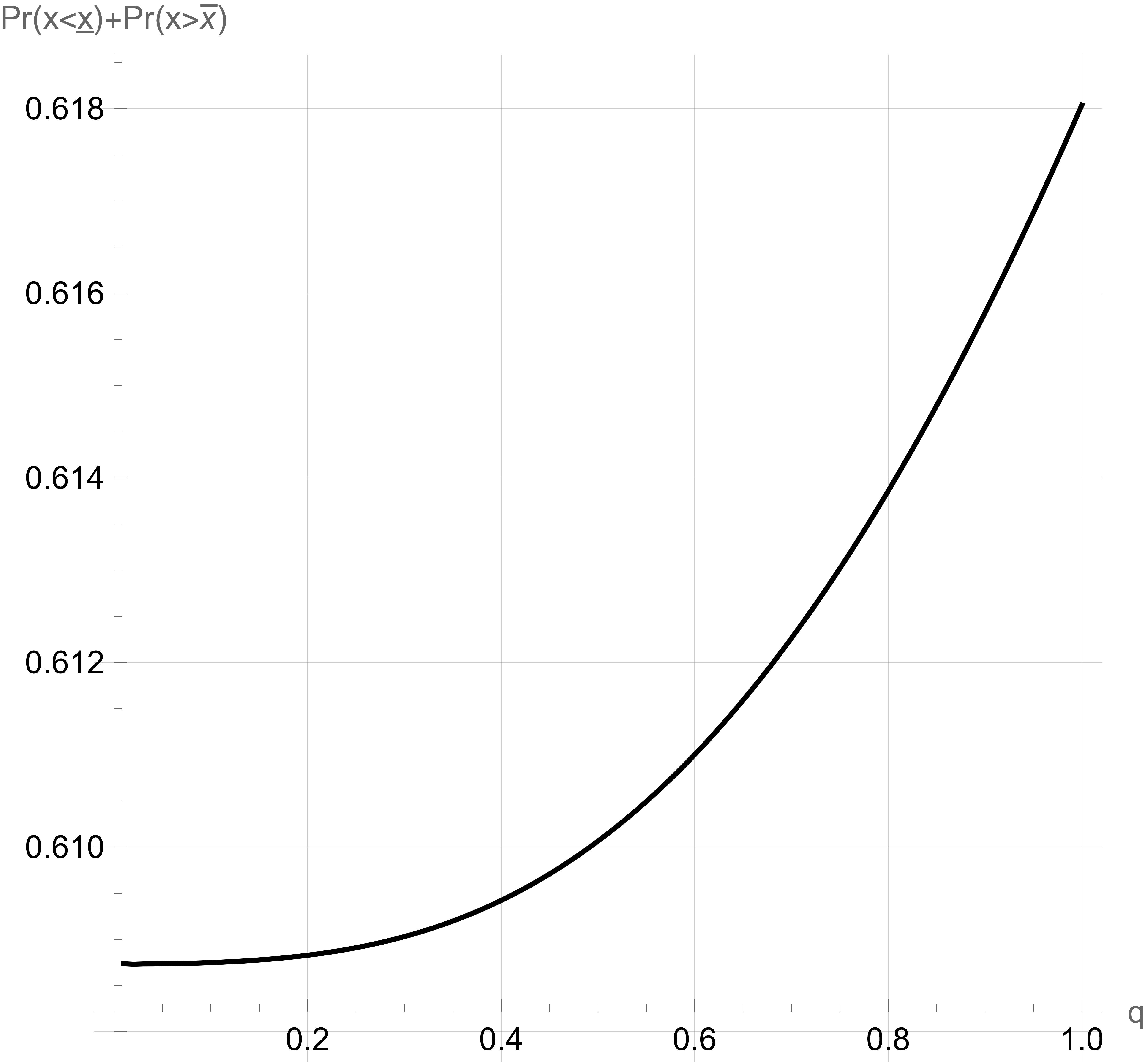}
		\caption{\footnotesize Parameters: $\beta=0.5$, $r_0=1$, $r=0.5$, $\mu_0=1$, $\sigma_x=1/2$}
		\label{q_numerical1}
	\end{subfigure}
	\hfill
	\begin{subfigure}[b]{0.45\textwidth}
		\centering
		\includegraphics[width=\textwidth]{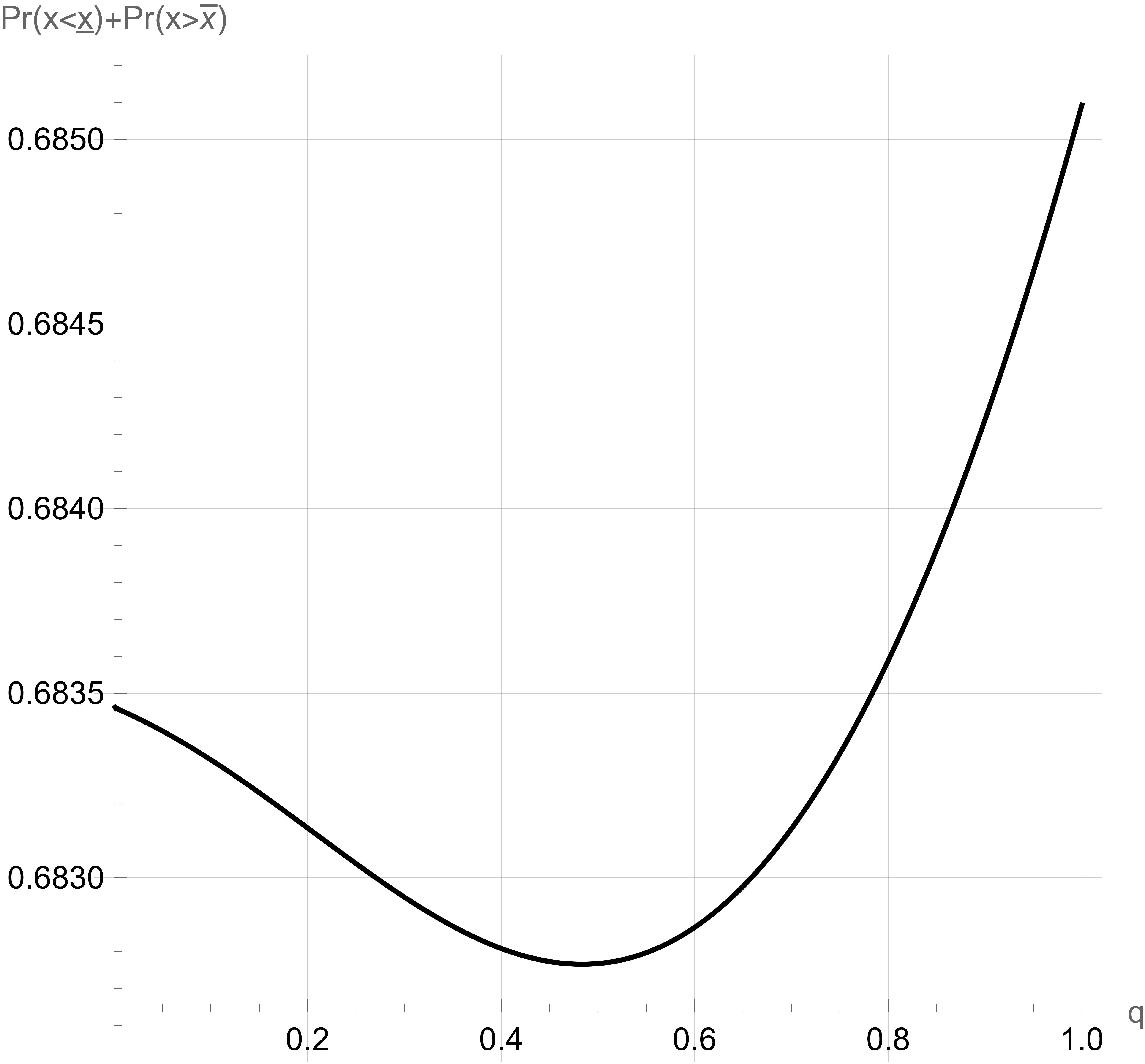}
		\caption{\footnotesize Parameters: $\beta=0.7$, $r_0=1$, $r=0.5$, $\mu_0=1$, $\sigma_x=1/2$}
		\label{q_numerical2}
	\end{subfigure}
	\caption{\textbf{The elaborateness of reports as a function of $q$}}
	\label{fig:q}
\end{figure}

The next corollary provides predictions on misleading disclosure, market reaction, and the report's extremity.
\begin{corollary}
	\label{maincs2}
	\begin{enumerate}
		\item The likelihood of misleading disclosure decreases with the investor's competence and the firm's information environment.

		\item Holding constant the initial evidence, the stock price reaction
		to a simple report (or equivalently, price volatility following a simple report) decreases with the investor's competence and increases with the firm's information environment.
		
		\item The extremity of elaborate reports decreases with
		the investor's competence. 
	
		\item When $r<r_0$, the extremity of negative (positive) reports increases (decreases)
		with the firm's information environment. When $r>r_0$, the extremity of positive (negative) reports increases (decreases) with the firm's information environment.
	\end{enumerate}
\end{corollary}

To understand Part 1 of Corollary \ref{maincs}, first recall from Proposition \ref{maincs} that as $\beta$ increases, the non-disclosure region shrinks. As a result, the scope for misleading disclosure is restricted. In addition, when $q$ increases, an investor's misleading disclosure is more likely to be confronted by other sources, and thus the investor issues less misleading disclosure.\footnote{However, this does not necessarily imply an increase in price informativeness. As shown in Proposition \ref{maincs}, when $q$ increases, the investor also withholds more additional evidence that aligns with the direction of the initial evidence, which reduces price informativeness. Therefore, the overall effect of the firm's information environment on price informativeness remains ambiguous.}

Part 2 of Corollary \ref{maincs} is perhaps surprising since one might expect that market reaction would increase with competence. However, this intuition may not be true after considering the investor's endogenous disclosure strategy. More successful hedge funds are more able to produce elaborate reports, which makes their simple reports more suspicious, and thus, the market reaction decreases with competence. Part 2 also links the market's reaction with the firm's information environment. This relationship would be ambiguous without solving for the investor's endogenous disclosure strategy. 

Since elaborate reports become less extreme when they become more frequent, Part 3 of Corollary \ref{maincs} follows directly from Proposition \ref{maincs}. For instance, in Figures \ref{equilibrium2} and \ref{equilibrium3}, as the size of region A becomes larger, there are more elaborate negative reports. However, as the size of region A increases and more elaborate negative reports are disclosed, the disclosed value of $x$ becomes closer to point $O$ and thus less extreme. If we map these comparative statics to the real world and consider elaborate reports by hedge funds, those that are more successful should, on average, issue less aggressive reports.

Our analysis so far focuses on the effects of $\beta$ and $q$ on the disclosure thresholds ($\underline{x}$ and $\bar{x}$), but these thresholds are also functions of the initial evidence, $r$. Figure \ref{fig:r} illustrates how various dependent variables change with $r$, assuming $x$ is normally distributed with mean $\mu_0$ and variance $\sigma_x^2$. Panel (a) shows that the thresholds converge as $r$ increases for $r < r_0$ and diverge for $r > r_0$. Panel (b) reveals that the firm value at the upper (lower) threshold decreases (increases) with $r$ for $r < r_0$, and this relationship reverses for $r>r_0$. Panels (c) and (d) depict the non-disclosure price, which always increases with $r$ when $\sigma_x$ is small but may decrease with small $r$ when $\sigma_x$ is large.

\begin{figure}[H]
	\centering
	\begin{subfigure}[b]{0.45\textwidth}
		\centering
		\includegraphics[width=\textwidth]{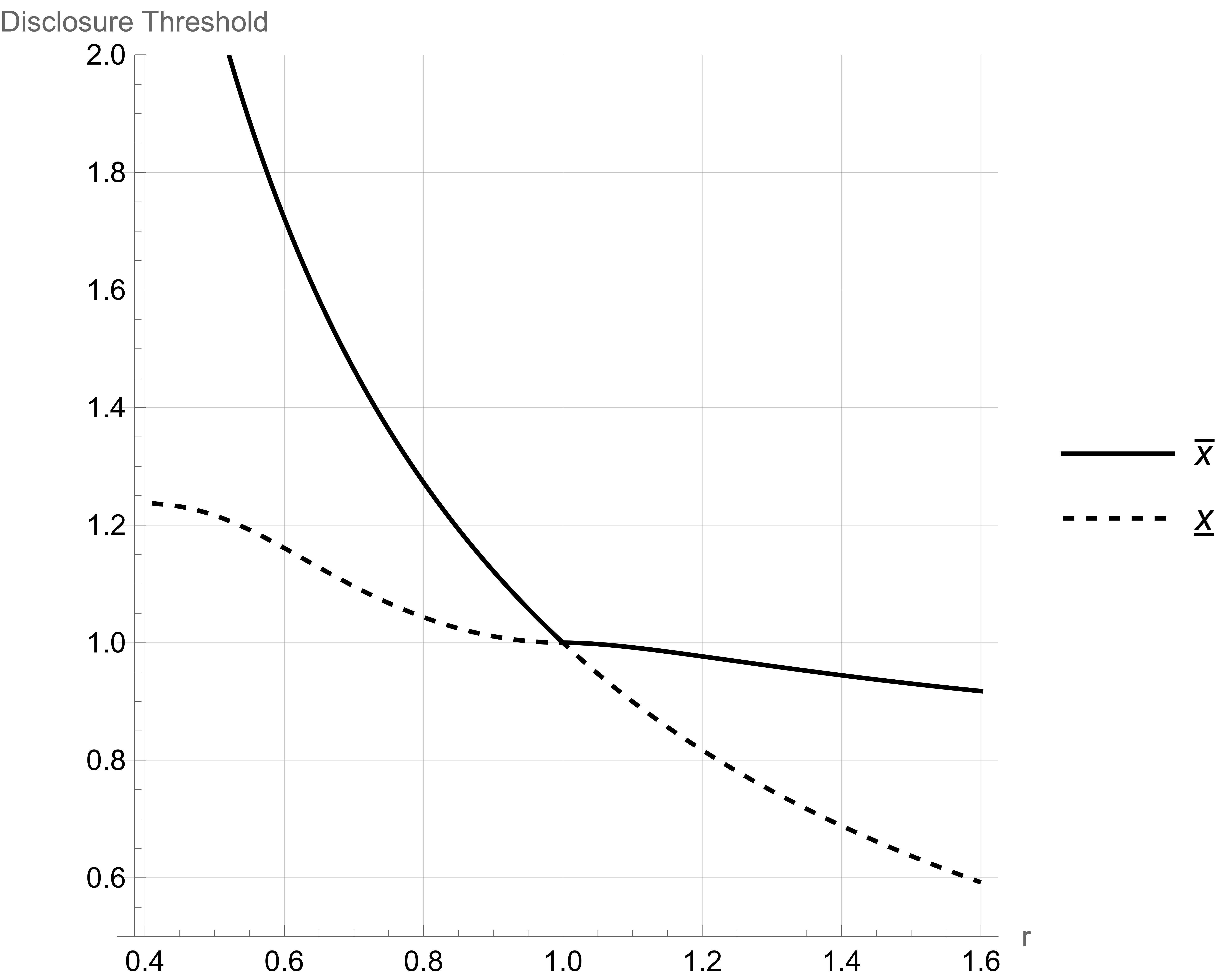} 
		\caption{\footnotesize The disclosure thresholds as a function of $r$, $\beta=0.7$, $q=0.8$, $r_0=1$, $\mu_0=1$, $\sigma_x=1/2$}
		\label{fig:sub1}
	\end{subfigure}
	\hfill
	\begin{subfigure}[b]{0.45\textwidth}
		\centering
		\includegraphics[width=\textwidth]{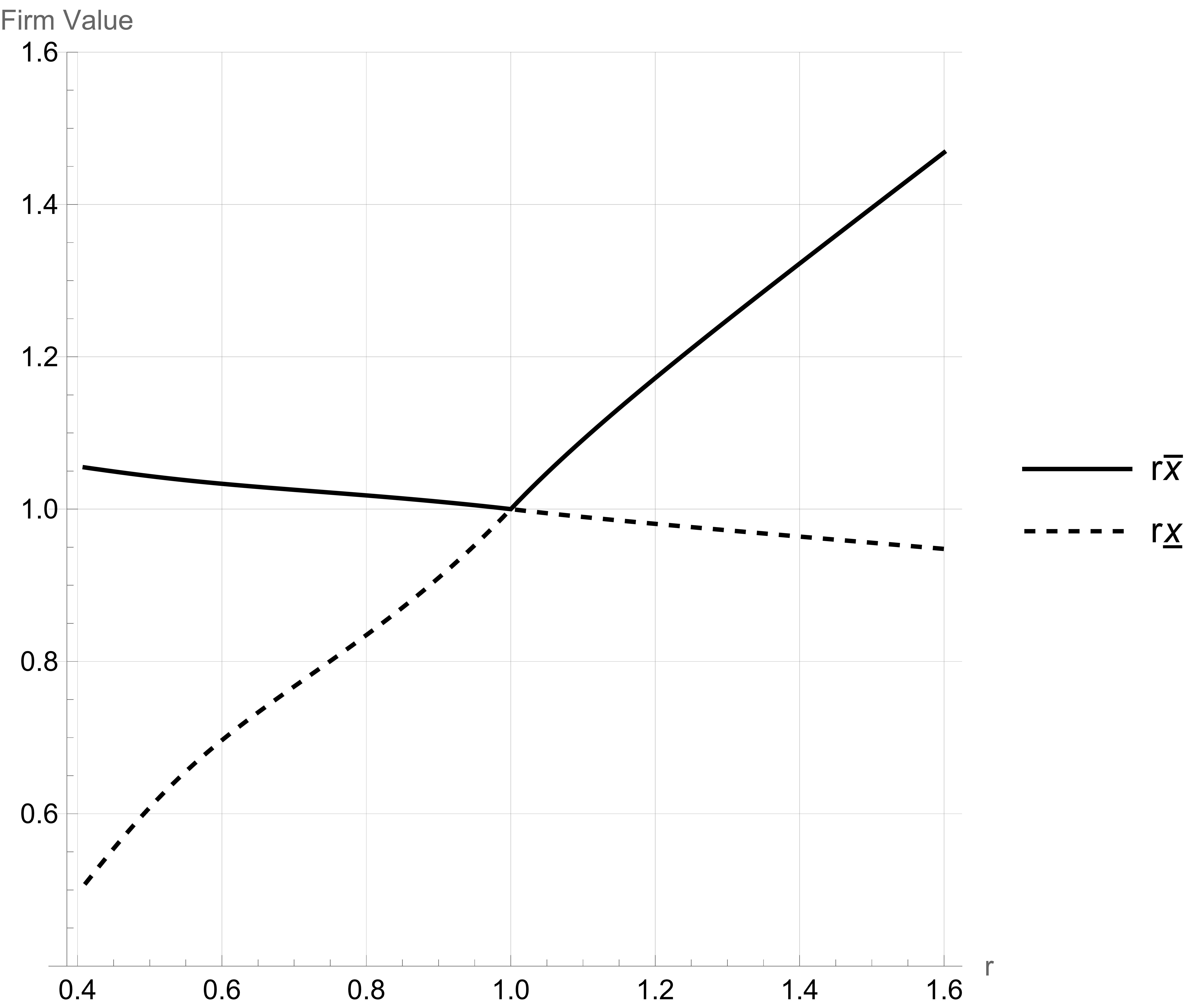} 
		\caption{\footnotesize Firm value at the disclosure thresholds as a function of $r$, $\beta=0.7$, $q=0.8$, $r_0=1$, $\mu_0=1$, $\sigma_x=1/2$}
		\label{fig:sub2}
	\end{subfigure}
	
	\vspace{0.5cm} 
	
	\begin{subfigure}[b]{0.45\textwidth}
		\centering
		\includegraphics[width=\textwidth]{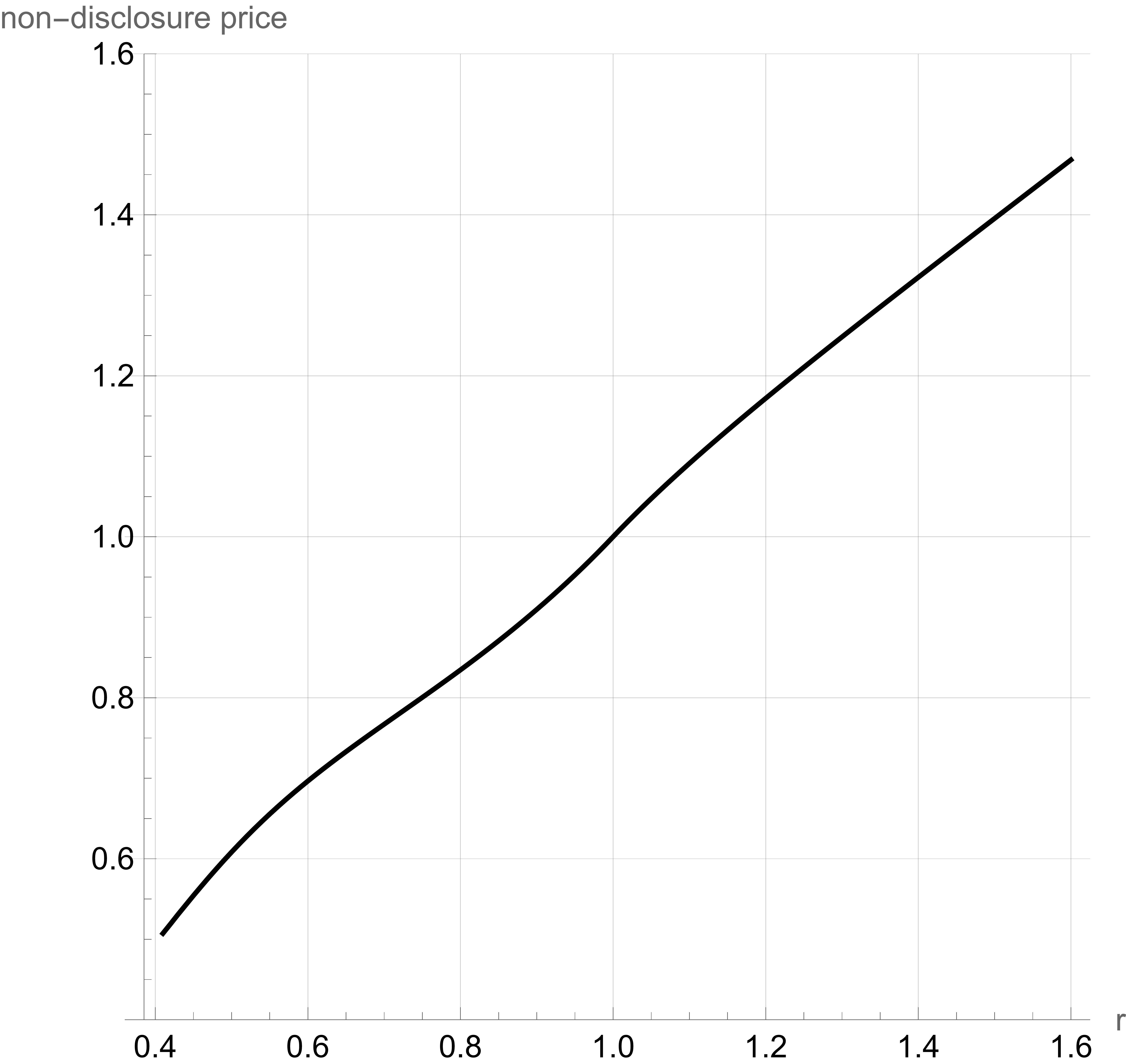} 
		\caption{\footnotesize Non-disclosure price as a function of $r$, $\beta=0.7$, $q=0.8$, $r_0=1$, $\mu_0=1$, $\sigma_x=1/2$}
		\label{fig:sub3}
	\end{subfigure}
	\hfill
	\begin{subfigure}[b]{0.45\textwidth}
		\centering
		\includegraphics[width=\textwidth]{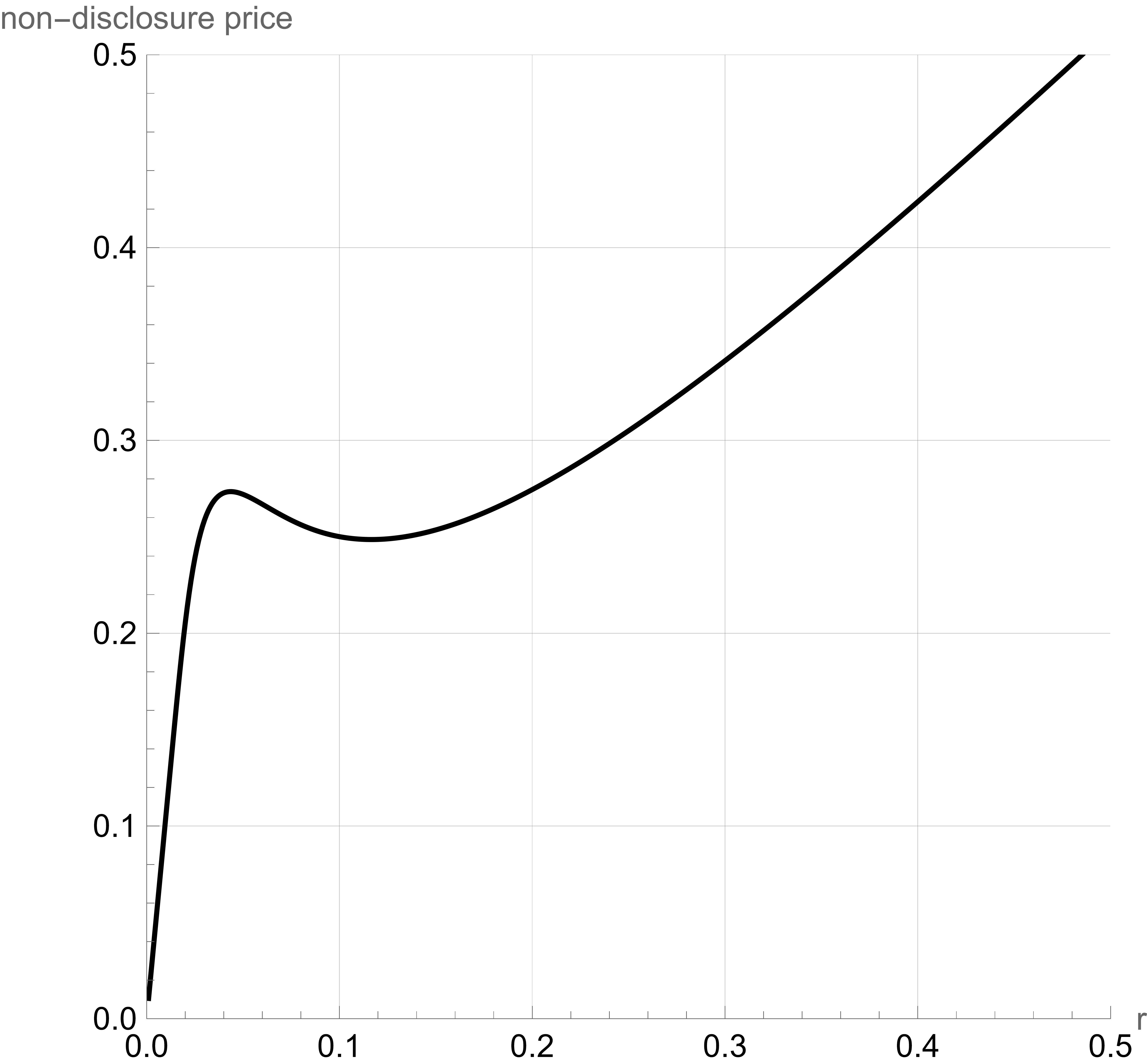} 
		\caption{\footnotesize Non-disclosure price as a function of $r$, $\beta=0.7$, $q=0.8$, $r_0=1$, $\mu_0=1$, $\sigma_x=20$}
		\label{fig:sub4}
	\end{subfigure}
	\caption{\textbf{Key variables as functions of $r$}}
	\label{fig:r}
\end{figure}

\section{Extension}\label{Section: Extension}
In this section, we extend our baseline model to consider the target firm's endogenous response and reexamine the investor's strategies.

We assume that after the investor issues a simple report, the target firm decides whether to respond. Specifically, the target firm is informed about the additional evidence $x$ with probability $p\in(0,1)$. The informed type can truthfully disclose $x$ or stay silent. The uninformed type can only stay silent. The firm's objective is to maximize the date 3 stock price.\footnote{An alternative approach to model the firm's response is to assume the firm can acquire $x$ with probability $1$ at a cost $c>0$. In this scenario, the firm will never acquire $x$. This is because, conditional on the investor’s non-disclosure of $x$, the firm’s expectation of $x$ is the same as the market’s (i.e., $\mu_{ND}$). Therefore, from the firm’s perspective, the expected price will remain the same if it acquires and discloses $x$ (i.e., $E[P_2|m=(r,\emptyset)]=E[rx|m=(r,\emptyset)]=r\mu_{ND}$). As a result, none of our baseline results change. For this reason, we choose to model the firm's response under uncertainty about information endowment (\citealt{dye1985disclosure}).}

The timeline is as follows:
\begin{itemize}
	\item At date 1, the investor learns $r$ with probability $\alpha$. Conditional on learning $r$, the investor learns $x$ with probability $\beta$. After this information event, the investor establishes an initial position $\rho $ at price $P_{1}=E[\theta]=r_0\mu_0$.
	
	\item At date 2, the investor makes the disclosure decisions. After the investor issues a simple report, the informed firm decides whether to disclose $x$.
	
	\item At date 3, the investor closes their initial position at the updated stock
	price $P_{2}$.
\end{itemize}

The next proposition presents our main results for this extension.
\begin{proposition}\label{prop:firm response}
	There exists an equilibrium in which:
	\begin{enumerate}
		\item The informed investor and the partially informed investor always disclose the initial evidence $r$. 
		
		\item  For any realization of $r$, the informed investor's unique equilibrium disclosure and trading decision are characterized by three regions:
		\begin{enumerate}
			\item	For $x\in (-\infty,\underline{x}_p)$, the investor takes a short position and discloses $x$.
			
			\item   For $x\in (\bar{x}_p,+\infty)$, the investor takes a long position and discloses $x$.
			
			\item	For $x\in [\underline{x}_p,\bar{x}_p]$, the investor withholds $x$ and takes a short (long) position when $r<\bar{r}$ ($r>\bar{r}$).
		\end{enumerate}
The partially informed investor takes a short (long) position when $r<\bar{r}$ ($r>\bar{r}$). The thresholds $\underline{x}_p$, $\bar{x}_p$, and $\bar{r}$ are defined in the appendix.

	\item When $r<\bar{r}$ ($r>\bar{r}$), the informed firm responds to the investor's simple report if and only if $x>\underline{x}_p$ ($x>\bar{x}_p$). 
	\end{enumerate}
\end{proposition}

Proposition \ref{prop:firm response} shows that our main results regarding the investor's disclosure strategy remain qualitatively unchanged after considering the firm's endogenous response. Similar to Proposition \ref{THM: r-disclosure_continuous} in the baseline model, Part 1 shows that the investor always disclosing $r$ is still an equilibrium. Even though the firm's response may dilute the informed investor's profit from withholding $x$, the informed investor can always guarantee positive profit by disclosing both $r$ and $x$. Thus, the full disclosure of $r$ is an equilibrium. In this equilibrium, the uninformed investor does not trade, as the price does not move upon non-disclosure of $r$.
 
Similar to Proposition \ref{x disclosure} in the baseline model, Part 2 demonstrates that the investor discloses extreme news and withholds moderate news. In equilibrium, the market's belief about $x$ conditional on its non-disclosure, $\mu_{ND}$, is a constant. Consequently, due to trading flexibility and for the same reasons outlined in the baseline model, the investor will disclose sufficiently high or low values of $x$ when the price volatility from disclosing $x$ exceeds that from withholding it. Compared with the baseline model where information disclosure by other sources is exogenous, the non-disclosure of $x$ can also indicate the firm's withholding of $x$. Thus, $\mu_{ND}$ is different from the baseline. As a result, the thresholds $\underline{x}_p$, $\bar{x}_p$, and $\bar{r}$ are also different from the baseline due to the firm's endogenous response. 

Since the firm's objective is to maximize the price, it will disclose $x$ if and only if $x$ is above $\mu_{ND}$. When the initial evidence is sufficiently unfavorable (i.e., when $r<\bar{r}$), the non-disclosure price is given by $r\mu_{ND}=r\underline{x}_p$, and the firm's disclosure threshold coincides with the investor's lower disclosure threshold. In this case, the firm aims to increase the stock price by withholding $x$, and the investor seeks to decrease the price and profit from the short position by withholding $x$. Consequently, when the investor withholds $x$, the informed firm will always disclose $x$ to increase the stock price. 

When the initial evidence is sufficiently favorable (i.e., when $r>\bar{r}$), the non-disclosure price is given by $r\mu_{ND}=r\bar{x}_p$, and the firm's disclosure threshold coincides with the investor's upper disclosure threshold. In this case, when the investor withholds $x$, the informed firm will not respond, as such withholding aligns with the firm's incentive to maximize stock price. However, when the investor is uninformed about $x$ and issues an honest simple report, the informed firm will disclose high values of $x$ to further increase the stock price.

The next corollary examines misleading disclosure and the predictions on the investor's competence. 
\begin{corollary}
	\label{cs3}
	\begin{enumerate}
\item Misleading disclosure arises when $x\in[\frac{r_0\mu_0}{r},\bar{x}_p]$ for $r<\bar{r}$ and when $x\in [\underline{x}_p,\frac{r_0\mu_0}{r}]$ for $r>\bar{r}$.  

\item The elaborateness of reports increases with the investor's competence. That is, $\frac{\partial \bar{x}_p}{\partial \beta}<0$ and $\frac{\partial \underline{x}_p}{\partial\beta}>0$.	
	 \end{enumerate}
\end{corollary}

Part 1 of Corollary \ref{cs3} shows that misleading disclosure still arises in equilibrium. Part 2 shows that the empirical predictions in Corollary \ref{maincs2} regarding the investor's competence remain qualitatively unchanged after incorporating the firm's endogenous response.

\section{Conclusions \label{Section: Conclusion}}
Our model suggests that misleading disclosure can occur when investors selectively withhold multi-dimensional information. This practice is more defensible legally than spreading unfounded rumors or disseminating fabricated reports, as it is difficult for regulators to prove that investors have deliberately withheld evidence. Therefore, we believe that this practice is more widespread than those involving misinformation and rumors. While the model's insights primarily apply to investors with direct trading incentives, such as hedge funds, they also apply to those with incentives to indirectly generate trading volume. For instance, considering that sell-side analysts are compensated based on the trading volume generated by their reports (e.g., \citealt{jackson2005trade} and \citealt{karmaziene2023greater}), our results suggest that they may also engage in misleading disclosure.

Our study has the following limitations. First, we assume the investor has a short horizon and closes their position after disclosure. If the investor can trade multiple times, they may engage in dynamic disclosure management. For instance, the investor can first disclose the decline in items per order, buy at depressed prices, and then resell after disclosing that the decline in items per order is actually good news. We examine this possibility in Appendix S1.

Second, our analysis focuses solely on the price impact of the investor's disclosure and does not model the price impact of the investor's trading. Incorporating the investor's multi-dimensional information into a trading game, such as \cite{kyle85e} and \cite{glosten1985bid}, would render our model intractable. Therefore, we present our analysis, specifically the unraveling result in Proposition \ref{THM: r-disclosure_continuous}, as a benchmark. An interesting avenue for future research is to examine whether and how the investor's disclosure incentives will be affected by the price impact of trading. 

Third, we only model investors' verifiable disclosure. Future research can examine alternative situations where investors' disclosure is partially unverifiable. For example, one can consider a scenario similar to \cite{dutta2002interpretation}, where the initial evidence $r$ is hard information that can be credibly disclosed, whereas the additional signal $x$ is soft information that cannot be credibly disclosed.

Finally, we do not investigate how investors’ disclosure affects firms' behavior. For example, whereas activist short-selling may deter firms from engaging in earnings management, it may also disrupt firms' normal communication using discretionary accruals, because discretionary accruals can be seen as red flags for earnings management. How does investors’ disclosure affect the overall amount of information reflected in stock prices? Can the benefits of activist short-selling be outweighed by the costs? We leave these questions to future research.

\newpage
\appendix
\section*{Appendix}
\subsection*{Proof of Proposition \ref{THM: r-disclosure_continuous}}
\subsubsection*{Step 0: Unraveling in the single signal case}
Before we prove Proposition \ref{THM: r-disclosure_continuous}, we prove a preliminary unraveling result in the case of a single signal. The proof of this result can help readers better understand the logic behind the proof of Proposition \ref{THM: r-disclosure_continuous}. 

Consider a simplified version of the baseline model where $r=1$ is common knowledge, so that $\theta=x$. At date 1, the investor learns $x$ with probability $\beta$ and establishes an initial position at price $P_{1}=E[x]=\mu_0$. At date 2, the investor makes the disclosure decision about $x$. At date 3, the investor closes their initial position at the updated stock price $P_{2}$.
\begin{lemma} \label{THM: r-disclosure_single}
	The investor always discloses $x$, regardless of its realization. 
\end{lemma}
\begin{proof}
We prove Lemma \ref{THM: r-disclosure_single} by contradiction. Suppose there exists a non-zero measure set $W$ such that the investor withholds $x$ when $x\in W$. We use $U$ to denote the event in which the investor is uninformed about $x$. We use $ND\equiv W\cup U$ to denote the event of non-disclosure. The date 3 price conditional on non-disclosure, which we denote by $\mu_{ND}$, is therefore 
	\begin{align*}
		\mu_{ND}&=\Pr (W|ND)E[x |W]+\Pr (U|ND)E[x|U] \\
		& =\Pr (W|ND)E[x|W]+(1-\Pr (W|ND))\mu_0.
	\end{align*}
Recall that the investor maximizes price volatility when making the disclosure decision. When the investor withholds $x$, the price volatility is $v_{ND}\equiv|\mu_{ND}-\mu_0|$. Since $\mu_{ND}$ is a weighted average of $E[x|W]$ and $\mu_0$, we have $|E[x|W]-\mu_0|>v_{ND}$. Furthermore, because $W$ has non-zero measure, there exists a type $x_l\in W$ such that $|x_l-\mu_0|>|E[x|W]-\mu_0|>v_{ND}$. However, this inequality suggests that the investor can induce a higher price volatility by disclosing $x_l$, which contradicts $x_l\in W$.
\end{proof}

\subsubsection*{Step 1: Preparations}
We now proceed to the proof of Proposition \ref{THM: r-disclosure_continuous}. We first introduce some notations. We use $T$ to denote the set of all possible types: $T=\{(r,x),(r,\emptyset),(\emptyset,\emptyset)|r\in \mathbb{R^{+}}, x\in \mathbb{R}\}$, where $(r,\emptyset)$ means the investor is informed about $r$ and uninformed about $x$, and $(\emptyset,\emptyset)$ means the investor is uninformed about both $r$ and $x$. We use $d$ to denote the investor's disclosure strategy, which maps the investor's type $t\in T$ to a reported message $m\in T$. The informed types can imitate the uninformed types but not vice-versa. Therefore, the possible strategies are: $d(\emptyset,\emptyset)=(\emptyset,\emptyset)$, $d(r,\emptyset)\in\{(r,\emptyset),(\emptyset,\emptyset)\}$, and $d(r,x)\in\{(r,x),(r,\emptyset),(\emptyset,\emptyset)\}$.\footnote{When the investor withholds $r$, they must also withhold $x$ because the disclosure of $x$ reveals that the investor is informed about $r$.} 

The equilibrium concept we use is Grossman–Perry–Farrell equilibrium (GPFE) defined in \cite{bertomeu2018verifiable}, which is a refinement of perfect Bayesian equilibrium (PBE). A strategy profile $\zeta$ is a PBE if: (i) The investor's disclosure and trading strategy maximizes their expected profit, given the market's pricing function $P_2(.)$; (ii) For any message $m\in d(T)$, the market's pricing function satisfies: $P_2(m)=E[\theta|t\in d^{-1}(m)]$.

PBE does not put restrictions on off-equilibrium beliefs. In our game, the message $(r,\emptyset)$ can be off-equilibrium, which requires us to apply equilibrium refinement techniques to pin down the market's belief. Definition \ref{self1} adapts the definition of self-signaling set from \cite{bertomeu2018verifiable} to our context.
\begin{definition}\label{self1}
	A non-empty set $\chi$ of types is \textit{self-signaling} relative to $\zeta$ if there exists an off-equilibrium message $m^*$ such that: (i) Conditional on receiving $m^*$, the market believes that the investor's type is in $\chi$ (i.e., $P_2(m^*)=E[\theta|t\in\chi]$). (ii) Types in $\chi$ strictly benefit from sending $m^*$. (iii) Types not in $\chi$ do not benefit from sending $m^*$. 
\end{definition}
The idea of Definition \ref{self1} is that an equilibrium is not sensible if some types can send a ``convincing'' off-equilibrium message: The types sending the message become strictly better-off if the market holds the correct belief that the message comes from these types. We say $\zeta$ is a GPFE if it is a PBE and there is no self-signaling set. 

We use $\theta_{ND}$ to denote the market's posterior expectation of $\theta$ conditional on non-disclosure (i.e., when there is no disclosure by the investor and also $\theta$ is not revealed by other sources). We assume $\theta_{ND}<r_0\mu_0$ without loss of generality.\footnote{The proof for the case $\theta_{ND}\ge r_0\mu_0$ is completely symmetrical.} We use $W_r$ to denote the set of withheld values of $r$: $W_r=\{r|d(r,x)=(\emptyset,\emptyset) \text{ for some $x$}, \text{ or } d(r,\emptyset)=(\emptyset,\emptyset)\}$. We further decompose $W_r$ into two mutually exclusive subsets, based on whether $d(r,\emptyset)=(\emptyset,\emptyset)$: $W_{a}=\{r|d(r,\emptyset)=(r,\emptyset), \text{ and }d(r,x)=(\emptyset,\emptyset) \text{ for some $x$}\}$; $W_{b}=\{r|d(r,\emptyset)=(\emptyset,\emptyset)\}$.\footnote{Note that $W_a\cup W_b=W_r$ and $W_a\cap W_b=\emptyset$.}

We next establish a lemma that pins down the investor's disclosure strategy when they are informed about both $r$ and $x$.
\begin{lemma}
	\label{short}
	If the investor is informed about both $r$ and $x$ and $d(r,x)=(\emptyset,\emptyset)$, the investor takes a short position. 
\end{lemma}
\begin{proof}
	We prove Lemma \ref{short} by contradiction. Suppose there exists a type $(r,x)$, who chooses $d(r,x)=(\emptyset,\emptyset)$ and takes a long position. This type expects the date 3 price to be:
	\begin{align} \label{p2zeta}
		E[P_2|t=(r,x),m=(\emptyset,\emptyset)]=(1-q)\theta_{ND}+qrx.
	\end{align}
	Since the investor takes a long position, the expected date 3 price should be higher than the prior: $E[P_2|t=(r,x),m=(\emptyset,\emptyset)]=(1-q)\theta_{ND}+qrx\ge P_1=r_0\mu_0$. Because $\theta_{ND}<r_0\mu_0$, we have $rx>r_0\mu_0>\theta_{ND}$. However, if the investor discloses both $r$ and $x$, the date 3 price will be $rx$, which is strictly higher than $E[P_2|t=(r,x),m=(\emptyset,\emptyset)]$. We reach a contradiction because the investor can induce a higher price volatility (i.e., increase their profit) by disclosing both $r$ and $x$. Therefore, the investor must take a short position. 
\end{proof}

Having finished the preparations, we prove Proposition \ref{THM: r-disclosure_continuous} by contradiction. Suppose a non-zero measure of types withhold $r$ in equilibrium (i.e., $W_r$ has non-zero measure). We will proceed in several steps. In Step 2, we will reach a contradiction and prove Proposition \ref{THM: r-disclosure_continuous} under the assumption that $W_b$ is empty. In Step 3, we establish some facts about the investor's disclosure strategy when $r\in W_b$. In Step 4, we use Bayes' rule to calculate $\theta_{ND}$ and derive conditions for the equilibrium disclosure thresholds. In Step 5, we apply Definition \ref{self1} and show that $W_b$ is empty if the equilibrium survives the self-signaling test, thus completing the proof.

\subsubsection*{Step 2: Reach a contradiction assuming $W_b$ is empty}
Step 2 is similar to Step 0: We will show that some withholding types can make a higher profit by disclosing both $r$ and $x$. 

Use $A(r)$ to denote the set of $x$ in which the investor withholds both $r$ and $x$ when $r\in W_a$. If $W_b=\emptyset$, we can calculate $\theta_{ND}$ using Bayes' rule:
\begin{align*}
	\theta_{ND}=\frac{\alpha\beta \int_{r\in W_a}f(r)(\int_{\{x|x\in A(r)\}}g(x)dx)dr}{P}E[\theta|r\in W_a,x\in A(r)]+\frac{1-\alpha}{P}\mu_0r_0,
\end{align*}
where $P=1-\alpha+\alpha\beta \int_{r\in W_a}f(r)(\int_{\{x|x\in A(r)\}}g(x)dx)dr$, which is the total probability that there is no disclosure by the investor. Because $\theta_{ND}<r_0\mu_0$, we know $E[\theta|r\in W_a,x\in A(r)]<\theta_{ND}<r_0\mu_0$. If $W_a$ has non-zero measure, there exists a type $(r_s,x_s)$ such that $r_s\in W_a$, $x_s\in A(r_s)$, and $r_sx_s<E[\theta|r\in W_a,x\in A(r)]$. However, if this type deviates and disclose both $r_s$ and $x_s$, the expected date 3 price is $E[P_2|t=(r_s,x_s),m=(r_s,x_s)]=r_sx_s<(1-q)r_sx_s+q\theta_{ND}=E[P_2|t=(r_s,x_s),m=(\emptyset,\emptyset)]$, which implies price volatility (i.e., the profit from the short position) will become strictly larger. Thus, $W_a$  must have zero measure, and our initial premise that $W_r$ has non-zero measure must be false, which implies the investor will always disclose $r$. 

From this point forward, we will complete the proof by showing that $W_b$ is empty.

\subsubsection*{Step 3: The disclosure strategy when $r\in W_b$}
The definition of $W_b$ does not pin down the disclosure strategy of types who are informed about both $r$ and $x$ when $r\in W_b$. The following lemma rules out the possibility $d(r,x)= (r,\emptyset)$ when $r\in W_b$.
\begin{lemma}
	\label{unravel}
$d(r,x)\neq (r,\emptyset)$ when $r\in W_b$.
\end{lemma}

The proof of Lemma \ref{unravel} is essentially the same as Step 0: Because $d(r,\emptyset)=(\emptyset,\emptyset)$ when $r\in W_b$, the market is \textit{certain} that the investor is informed about $x$ upon observing $m=(r,\emptyset)$ with $r\in W_b$. This implies that $d(r,x)\neq (r,\emptyset)$ when $r\in W_b$, as otherwise some types can induce a higher price volatility (i.e., make a higher profit) by disclosing both $r$ and $x$. For details of the proof, please see Section 1 of Appendix S1. 

Lemma \ref{unravel} implies that if the investor is informed about both $r$ and $x$ and if $r\in W_b$, the disclosure strategy only has two possibilities: $d(r,x)\in\{(r,x),(\emptyset,\emptyset)\}$ (i.e., the message $(r,\emptyset)$ is off-equilibrium when $r\in W_b$). We have established in Lemma \ref{short} that when $d(r,x)=(\emptyset,\emptyset)$, the investor must take a short position. Using $\pi_q(r,x)$ to denote their profit (price volatility) in this case, we have $\pi_q(r,x)=P_1-E[P_2|t=(r,x),m=(\emptyset,\emptyset)]=r_0\mu_0-((1-q)\theta_{ND}+qrx)$. When $d(r,x)=(r,x)$, the investor takes a short position when $rx<r_0\mu_0$ and a long position when $rx>r_0\mu_0$. Thus, their profit (price volatility) is $\pi(r,x)=|rx-r_0\mu_0|$. The investor compares $\pi(r,x)$ with $\pi_q(r,x)$ to determine their optimal strategy. It is straightforward that $\pi(r,x)\le\pi_q(r,x)$ if and only if $rx\in[\theta_{ND},\frac{2r_0\mu_0-(1-q)\theta_{ND}}{1+q}]\equiv[\underline{\theta},\bar{\theta}]$, which is an non-empty interval because $\theta_{ND}<r_0\mu_0$. Thus, when the investor is informed about both $r$ and $x$ and $r\in W_b$, they withhold both when $rx\in[\underline{\theta},\bar{\theta}]$, and discloses both when $rx\notin[\underline{\theta},\bar{\theta}]$. The disclosure thresholds are determined by the following two equations:
\begin{align}
	&(1+q)\bar{\theta}+(1-q)\underline{\theta}=2r_0\mu_0, \label{bartheta}\\
	&\underline{\theta}=\theta_{ND}. \label{underlinetheta}
\end{align}
\subsubsection*{Step 4: Calculating $\theta_{ND}$}
Define 
\begin{align*}
	P&\equiv1-\alpha+\alpha(1-\beta)\int_{r\in W_b}f(r)dr+\alpha\beta \int_{r\in W_b}f(r)(\int_{\{x|rx\in[\underline{\theta},\bar{\theta}]\}}g(x)dx)dr\\
	&+\alpha\beta \int_{r\in W_a}f(r)(\int_{\{x|x\in A(r)\}}g(x)dx)dr,
\end{align*}
which is the total probability that there is no disclosure by the investor.

We can calculate $\theta_{ND}$ using Bayes' rule:
\begin{align}
	\theta_{ND}&=\frac{\alpha\beta \int_{r\in W_a}rf(r)(\int_{\{x|x\in A(r)\}}xg(x)dx)dr}{P}\notag\\
	&+\frac{\alpha\beta \int_{r\in W_b}rf(r)(\int_{\{x|rx\in[\underline{\theta},\bar{\theta}]\}}xg(x)dx)dr}{P}\notag\\
	&+\frac{\alpha(1-\beta)\mu_0\int_{r\in W_b}rf(r)dr}{P}\notag\\
	&+\frac{1-\alpha}{P}\mu_0r_0,\label{thetand}
\end{align}

We further decompose $W_b$ into two subsets, based on whether $r$ is bigger or smaller than the prior mean: $W_B=\{r|d(r,\emptyset)=(\emptyset,\emptyset),r\ge r_0\}$ and  $W_S=\{r|d(r,\emptyset)=(\emptyset,\emptyset),0<r<r_0\}$. The subscripts ``B'' and ``S'' represent big and small, respectively. Note that when $W_b$ is non-empty, \eqref{underlinetheta} has to hold. Plugging \eqref{underlinetheta} into \eqref{thetand} and using a similar technique as in \cite{jung1988disclosure} to simplify (i.e., integration by parts) leads to 
\begin{align}
	&\int_{r\in W_B}rf(r)(\beta\int_{\frac{\underline{\theta}}{r}}^{\frac{\bar{\theta}}{r}}(G(x)-G(\frac{\bar{\theta}}{r}))dx-(1-\beta)(\mu_0-\frac{\underline{\theta}}{r}))dr\notag\\
	+&\int_{r\in W_S}rf(r)(\beta\int_{\frac{\underline{\theta}}{r}}^{\frac{\bar{\theta}}{r}}(G(x)-G(\frac{\bar{\theta}}{r}))dx-(1-\beta)(\mu_0-\frac{\underline{\theta}}{r}))dr\notag\\
	+&\beta(\underline{\theta}\int_{r\in W_a}f(r)(\int_{\{x|x\in A(r)\}}g(x)dx)dr-\int_{r\in W_a}rf(r)(\int_{\{x|x\in A(r)\}}xg(x)dx)dr)\notag\\
	-&\frac{1-\alpha}{\alpha}(r_0\mu_0-\underline{\theta})\notag\\
	=&0\label{indiffc}.
\end{align}
We use $\zeta$ to denote the strategy profile characterized by the disclosure thresholds $\underline{\theta}$ and $\bar{\theta}$ that satisfy \eqref{bartheta}, \eqref{underlinetheta}, and \eqref{indiffc}.

\subsubsection*{Step 5: Equilibrium refinement}
Finally, we demonstrate that if $\zeta$ passes the self-signaling test in Definition \ref{self1}, then \eqref{indiffc} cannot hold as all 4 terms on the left-hand-side of \eqref{indiffc} are negative. This, in turn, implies that $W_b=\emptyset$, since \eqref{indiffc} must hold if $W_b\neq\emptyset$. The construction of self-signaling set involves applying the equilibrium properties in Proposition \ref{x disclosure}. The details of this step can be found at Section 1 of Appendix S1. 

\subsection*{Proof of Proposition \ref{x disclosure} for the case $r<r_0$}
As established in Lemma \ref{volatility}, the investor withholds $x$ if and only if $\pi_q(x)\ge\pi(x)$ ($\pi_q(x)$ and $\pi(x)$ are defined in equations \eqref{piq} and \eqref{pix}). The simplification of this inequality requires us to determine the relation between $r\mu_{ND}$ and $r_0\mu_0$. It can be easily shown that when $r\mu_{ND}\ge r_0\mu_0$, the equilibrium conditions have no solution, implying that no equilibrium exists in this case.\footnote{The proof of this claim is at Section 2 of Appendix S1.} When $r\mu_{ND}<r_0\mu_0$, it is straightforward that $\pi_q(x)\ge\pi(x)$ if and only if $x\in[\underline{x},\bar{x}]$, where the thresholds are determined by:
\begin{align}
	&\underline{x}=\mu_{ND}, \label{lower}\\
	&\bar{x}=\frac{2r_0\mu_0-(1-q)r\mu_{ND}}{(1+q)r}.\label{upper}
\end{align}
The last piece to pin down the equilibrium is to use Bayes' rule to calculate $\mu_{ND}$ in equilibrium: 
\begin{equation}
	\mu_{ND}=\frac{1-\beta}{1-\beta+\beta\int_{\underline{x}}^{\bar{x}}g(x)dx}\mu_0+\frac{\beta\int_ {\underline{x}}^{\bar{x}}g(x)dx}{1-\beta+\beta\int_{\underline{x}}^{\bar{x}}g(x)dx}E[x|x\in [\underline{x},\bar{x}]]. \label{Bayes}
\end{equation}

We can rewrite \eqref{lower}, \eqref{upper}, and \eqref{Bayes} as
\begin{align}
	&\beta\int_{\underline{x}}^{\bar{x}}(G(x)-G(\bar{x}))dx-(1-\beta)(\mu_0 -\underline{x})=0, \label{shortindiff1}\\
	&(1+q)\bar{x}+(1-q)\underline{x}=\frac{2r_0\mu_0}{r}, \label{shortindiff2}
\end{align}
where \eqref{shortindiff1} follows from \eqref{lower} and \eqref{Bayes} by using a similar technique as in \cite{jung1988disclosure} (i.e., integration by parts), and \eqref{shortindiff2} follows from \eqref{lower} and \eqref{upper}. 

\eqref{shortindiff1} and \eqref{shortindiff2} can be transformed to a single condition on  $\underline{x}$: 
\begin{align}\label{r<r_0}
	Q(\underline{x})\equiv\beta\int_{\underline{x}}^{\frac{2r_0\mu_0}{(1+q)r}-\frac{1-q}{1+q}\underline{x}}(G(x)-G(\frac{2r_0\mu_0}{(1+q)r}-\frac{1-q}{1+q}\underline{x}))dx-(1-\beta)(\mu_0 -\underline{x})=0.
\end{align}
Since $Q'(\underline{x})=\beta(G(\bar{x})-G(\underline{x}))+\beta\int_{\underline{x}}^{\bar{x}}\frac{1-q}{1+q}g(\bar{x})dx+(1-\beta)>0$, $Q(.)$ is an increasing function of $\underline{x}$. Furthermore, $Q(\frac{r_0\mu_0}{r})=-(1-\beta)(\mu_0-\frac{r_0\mu_0}{r})>0$ and $Q(\mu_0)<0$. Thus, there exists a unique $\underline{x}\in(\mu_0,\frac{r_0\mu_0}{r})$ that solves \eqref{r<r_0}, which confirms that the equilibrium exists and is unique. 

According to equations \eqref{piq} and \eqref{pix}, $\pi_q(x)=0$ when $x=\frac{r_0\mu_0-(1-q)r\mu_{ND}}{qr}$, and $\pi(x)=0$ when $x=\frac{r_0\mu_0}{r}$. Since $r\mu_{ND}<r_0\mu_0$ in equilibrium, we have $\frac{r_0\mu_0-(1-q)r\mu_{ND}}{qr}>\frac{r_0\mu_0}{r}$, which implies that the two intersections of $\pi_q(x)$ and $\pi(x)$ are located at the downward sloping part of $\pi_q(x)$, as shown in Figure \ref{equilibrium2}. In other words, when $x\in[\underline{x},\bar{x}]$, we have $(1-q)r\mu_{ND}+qrx<r_0\mu_0$ so that the expected date 3 price is lower than the date 1 price. Therefore, the investor must take a short position in regions B and C.

\subsection*{Proof of Proposition \ref{x disclosure} for the case $r>r_0$}
Following a similar procedure as the case $r<r_0$, we can show that no equilibrium exists in the case $r\mu_{ND}\le r_0\mu_0$. When, $r\mu_{ND}>r_0\mu_0$, $\pi_q(x)\ge\pi(x)$ if and only if $x\in[\underline{x},\bar{x}]$, where the thresholds are determined by:
\begin{align}
	&\bar{x}=\mu_{ND}, \label{lowerother}\\
	&\underline{x}=\frac{2r_0\mu_0-(1-q)r\mu_{ND}}{(1+q)r}. \label{upperother} 
\end{align}
In addition, by Bayes' rule,
\begin{equation}
	\mu_{ND}=\frac{1-\beta}{1-\beta+\beta\int_{\underline{x}}^{\bar{x}}g(x)dx}\mu_0+\frac{\beta\int_ {\underline{x}}^{\bar{x}}g(x)dx}{1-\beta+\beta\int_{\underline{x}}^{\bar{x}}g(x)dx}E[x|x\in [\underline{x},\bar{x}]]. \label{Bayesother}
\end{equation}
\eqref{lowerother}-\eqref{Bayesother} simplify to:
\begin{align}
	&\beta\int_{\underline{x}}^{\bar{x}}(G(x)-G(\underline{x}))dx-(1-\beta)(\mu_0 -\bar{x})=0, \label{longindiff1} \\
	&(1-q)\bar{x}+(1+q)\underline{x}=\frac{2r_0\mu_0}{r}, \label{longindiff2} 
\end{align}

Combining \eqref{longindiff1} and \eqref{longindiff2} leads to
\begin{align}\label{r>r_0}
	R(\bar{x})\equiv\beta\int_{\frac{2r_0\mu_0}{(1+q)r}-\frac{1-q}{1+q}\bar{x}}^{\bar{x}}(G(x)-G(\frac{2r_0\mu_0}{(1+q)r}-\frac{1-q}{1+q}\bar{x}))dx-(1-\beta)(\mu_0-\bar{x})=0.
\end{align}
$R(.)$ is an increasing function of $\bar{x}$ as $R'(\bar{x})=\beta(G(\bar{x})-G(\underline{x}))+\beta\int_{\underline{x}}^{\bar{x}}\frac{1-q}{1+q}g(\underline{x})dx+(1-\beta)>0$. Furthermore, $R(\frac{r_0\mu_0}{r})<0$ and $R(\mu_0)>0$. Thus, there exists a unique $\bar{x}\in(\frac{r_0\mu_0}{r},\mu_0)$ that solves \eqref{r>r_0}, which confirms that the equilibrium exists and is unique. Finally, the proof for the investor taking a long position when issuing a simple report is similar to the case $r<r_0$ and is therefore omitted.

\subsection*{Proof of Proposition \ref{maincs} for the case $r<r_0$}
Let $K(\underline{x},\bar{x},\beta)\equiv\beta\int_{\underline{x}}^{\bar{x}}(G(x)-G(\bar{x}))dx-(1-\beta)(\mu_0 -\underline{x})=0$. 
Differentiating \eqref{shortindiff1} and \eqref{shortindiff2} with respect to $\beta$ and applying the Implicit Function Theorem, we have 
\begin{align*}
&\frac{\partial \bar{x}}{\partial \beta}=\frac{-K_{\beta}}{K_{\bar{x}}-\frac{1+q}{1-q}K_{\underline{x}}},\\
&\frac{\partial \underline{x}}{\partial \beta}=\frac{-K_{\beta}}{K_{\underline{x}}-\frac{1-q}{1+q}K_{\bar{x}}}.
\end{align*}
Since $K_{\underline{x}}>0$, $K_{\bar{x}}<0$, and $K_{\beta}<0$, we have $\frac{\partial \bar{x}}{\partial \beta}<0$ and $\frac{\partial \underline{x}}{\partial \beta}>0$.

Differentiating \eqref{shortindiff1} and \eqref{shortindiff2} with respect to $q$ and applying the Implicit Function Theorem, we have 
\begin{align*}
&\frac{\partial \bar{x}}{\partial q}=\frac{\underline{x}-\bar{x}}{1+q-(1-q)\frac{K_{\bar{x}}}{K_{\underline{x}}}},\\
&\frac{\partial \underline{x}}{\partial q}=\frac{\underline{x}-\bar{x}}{1-q-(1+q)\frac{K_{\underline{x}}}{K_{\bar{x}}}}.
\end{align*}
Since $K_{\underline{x}}>0$ and $K_{\bar{x}}<0$, we have $\frac{\partial \bar{x}}{\partial q}<0$ and $\frac{\partial \underline{x}}{\partial q}<0$.

\subsection*{Proof of Proposition \ref{maincs} for the case $r>r_0$}

Let $J(\underline{x},\bar{x},\beta)\equiv\beta\int_{\underline{x}}^{\bar{x}}(G(x)-G(\underline{x}))dx-(1-\beta)(\mu_0 -\bar{x})$. 
Differentiating \eqref{longindiff1} and \eqref{longindiff2} with respect to $\beta$ and applying the Implicit Function Theorem, we have 
\begin{align*}
&\frac{\partial \bar{x}}{\partial \beta}=\frac{-J_{\beta}}{J_{\bar{x}}-\frac{1-q}{1+q}J_{\underline{x}}},\\
&\frac{\partial\underline{x}}{\partial\beta}=\frac{-J_{\beta}}{J_{\underline{x}}-\frac{1+q}{1-q}J_{\bar{x}}}.
\end{align*}
Since $J_{\underline{x}}<0$, $J_{\bar{x}}>0$, and $J_{\beta}>0$, we have $\frac{\partial \bar{x}}{\partial \beta}<0$ and $\frac{\partial \underline{x}}{\partial\beta}>0$.

Differentiating \eqref{longindiff1} and \eqref{longindiff2} with respect to $q$ and applying the Implicit Function Theorem, we have
\begin{align*}
&\frac{\partial \bar{x}}{\partial q}=\frac{\bar{x}-\underline{x}}{1-q-(1+q)\frac{J_{\bar{x}}}{J_{\underline{x}}}},\\
&\frac{\partial \underline{x}}{\partial q}=\frac{\bar{x}-\underline{x}}{1+q-(1-q)\frac{J_{\underline{x}}}{J_{\bar{x}}}}.
\end{align*}
Since $J_{\underline{x}}<0$, $J_{\bar{x}}>0$, we have $\frac{\partial \bar{x}}{\partial q}>0$ and $\frac{\partial \underline{x}}{\partial q}>0$.

\subsection*{Proof of Proposition \ref{prop:firm response}}
We first prove that if $r$ is always disclosed, the investor's and firm's disclosure strategies of $x$ are characterized by Parts 2-4 of Proposition \ref{prop:firm response}. We then verify that the investor always disclosing $r$ is an equilibrium.

\subsubsection*{Step 1: The disclosure strategies of $x$ when $r<r_0$}
We still use $\mu_{ND}$ to denote the market's posterior expectation of $x$ upon non-disclosure of $x$.  Note that $\mu_{ND}$ is a constant in equilibrium. This means that the firm will disclose $x$ if and only if $x$ is above $\mu_{ND}$. 

Now consider the disclosure decision by the investor. Following a similar argument as Proposition \ref{x disclosure}, the investor takes a short position whenever issuing a simple report, and the equilibrium features a non-disclosure interval $[\underline{x}_p,\bar{x}_p]$. At the lower threshold $\underline{x}_p$, the investor is indifferent between shorting plus withholding and shorting plus disclosure. When the investor withhold $\underline{x}_p$, they understand that with some probability $w\in\{0,p\}$, $\underline{x}_p$ will be disclosed by the firm. Note that at this point, we only know $w$ is either $0$ or $p$, as the informed firm may withhold $\underline{x}_p$. Thus, the indifference condition implies that
\begin{align*}
(1-w)r\mu_{ND}+wr\underline{x}_p=r\underline{x}_p,
\end{align*}
which simplifies to
\begin{align}
	&\underline{x}_p=\mu_{ND}. \label{lowerxp}
\end{align}
Equation \eqref{lowerxp} implies that the investor's lower disclosure threshold coincides with the firm's disclosure threshold, as both are equal to $\mu_{ND}$. Due to this observation, we know $w=p$, as when the investor withholds $x$, they understand that the informed firm will disclose $x$ for sure to increase the stock price.

At the upper threshold $\bar{x}_p$, the investor is indifferent between shorting plus withholding and longing plus disclosure. This indifference condition implies that
\begin{align*}
	r_0\mu_0-((1-p)r\mu_{ND}+pr\bar{x}_p)=r\bar{x}_p-r_0\mu_0,
\end{align*}
which further implies
\begin{align}
	&\bar{x}_p=\frac{2r_0\mu_0-(1-p)r\mu_{ND}}{(1+p)r}.\label{upperxp}
\end{align}

Note that non-disclosure can come from one of the following three events: (i) Both the firm and the investor are uninformed; (ii) The firm is uninformed, the investor is informed, and $x\in[\underline{x}_p,\overline{x}_p]$; (iii) The investor is uninformed, the firm is informed, and $x\in(-\infty,\underline{x}_p)$. Using Bayes' rule, we can write $\mu_{ND}$ as:
\begin{align}
	\mu_{ND}=&\frac{(1-\beta)(1-p)}{(1-\beta)(1-p)+\beta(1-p)\int_{\underline{x}_p}^{\bar{x}_p}g(x)dx+p(1-\beta)\int_{-\infty}^{\underline{x}_p}g(x)dx}\mu_0+\notag\\ 
	&\frac{\beta(1-p)\int_{\underline{x}_p}^{\bar{x}_p}g(x)dx}{(1-\beta)(1-p)+\beta(1-p)\int_{\underline{x}_p}^{\bar{x}_p}g(x)dx+p(1-\beta)\int_{-\infty}^{\underline{x}_p}g(x)dx}E[x |x\in [\underline{x}_p,\bar{x}_p]]+\notag\\ 
	&\frac{p(1-\beta)\int_{-\infty}^{\underline{x}_p}g(x)dx}{(1-\beta)(1-p)+\beta(1-p)\int_{\underline{x}_p}^{\bar{x}_p}g(x)dx+p(1-\beta)\int_{-\infty}^{\underline{x}_p}g(x)dx}E[x |x\in (-\infty,\underline{x}_p)]. \label{muNDextension}	
\end{align}
We can simplify \eqref{lowerxp}, \eqref{upperxp}, and \eqref{muNDextension} to the following two conditions:
\begin{align}
	&\frac{\beta}{1-\beta}\int_{\underline{x}_p}^{\bar{x}_p}(G(x)-G(\bar{x}_p))dx+\frac{p}{1-p}\int_{-\infty}^{\underline{x}_p}G(x)dx-(\mu_0-\underline{x}_p)=0, \label{shortindiff_ex}\\
	&(1+p)\bar{x}_p+(1-p)\underline{x}_p=\frac{2r_0\mu_0}{r}, \label{upper_xp}
\end{align}
where \eqref{shortindiff_ex} follows from \eqref{lowerxp} and \eqref{muNDextension} by using a similar technique as in \cite{jung1988disclosure} (i.e., integration by parts), and \eqref{upper_xp} follows from \eqref{lowerxp} and \eqref{upperxp}. 

Plugging \eqref{upper_xp} into \eqref{shortindiff_ex} leads to 
\begin{align*}
H(\underline{x}_p)\equiv &\frac{\beta}{1-\beta}\int_{\underline{x}_p}^{\frac{\frac{2r_0\mu_0}{r}-(1-p)\underline{x}_p}{1+p}}(G(x)-G(\frac{\frac{2r_0\mu_0}{r}-(1-p)\underline{x}_p}{1+p}))dx\\&+\frac{p}{1-p}\int_{-\infty}^{\underline{x}_p}G(x)dx-(\mu_0-\underline{x}_p)=0.
\end{align*}
$H(\underline{x}_p)$ is obviously an increasing function of $\underline{x}_p$ as all three terms increase in $\underline{x}_p$. Furthermore, we have $H(-\infty)<0$ and $H(\frac{r_0\mu_0}{r})>0$. By the intermediate value theorem, there exists a unique $\underline{x}_p\in(-\infty,\frac{r_0\mu_0}{r})$ that solves $H(\underline{x}_p)=0$. Therefore, the equilibrium thresholds are uniquely determined by equations  \eqref{shortindiff_ex} and \eqref{upper_xp}. 

\subsubsection*{Step 2: The disclosure strategies of $x$ when $r>r_0$}
We first conjecture that there exists an equilibrium where the investor takes a long position when they issue a simple report. That is, $r\mu_{ND}>r_0\mu_0$. 

In this strategy profile, following a similar argument as Proposition \ref{x disclosure}, the investor takes a long position whenever issuing a simple report, and the equilibrium features a non-disclosure interval $[\underline{x}_p,\bar{x}_p]$. At the upper threshold $\bar{x}_p$, the investor is indifferent between longing plus withholding and longing plus disclosure. When the investor withhold $\bar{x}_p$, they understand that with some probability $w\in\{0,p\}$, $\bar{x}_p$ will be disclosed by the firm. Note that at this point, we only know $w$ is either $0$ or $p$, as the informed firm may withhold $\bar{x}_p$. Thus, the indifference condition implies that
\begin{align*}
	(1-w)r\mu_{ND}+wr\bar{x}_p=r\bar{x}_p,
\end{align*}
which simplifies to
\begin{align}
	&\bar{x}_p=\mu_{ND}. \label{xpupper}
\end{align}
Equation \eqref{xpupper} implies that the investor's upper disclosure threshold coincides with the firm's disclosure threshold, as both are equal to $\mu_{ND}$. Due to this observation, we know $w=0$, as when the investor withholds $x$, they understand that the informed firm will not respond.

At the lower disclosure threshold $\underline{x}_p$, the investor is indifferent between longing plus withholding and shorting plus disclosure. This indifference condition implies that
\begin{align*}
	r\mu_{ND}-r_0\mu_0=r_0\mu_0-r\underline{x}_p,
\end{align*}
which further implies
\begin{align}
	&\underline{x}_p=\frac{2r_0\mu_0-r\mu_{ND}}{r}.\label{xplower}
\end{align}

Note that non-disclosure of $x$ can come from one of the following four events: (i) Both the firm and the investor are uninformed; (ii) The firm is uninformed, the investor is informed, and $x\in[\underline{x}_p,\bar{x}_p]$; (iii) The investor is uninformed, the firm is informed, and $x\in(-\infty,\bar{x}_p)$; (iv) Both the firm and the investor are informed, and $x\in[\underline{x}_p,\bar{x}_p]$. Using Bayes' rule, we can write $\mu_{ND}$ as:
\begin{align}
	\mu_{ND}=&\frac{(1-\beta)(1-p)}{(1-\beta)(1-p)+\beta\int_{\underline{x}_p}^{\bar{x}_p}g(x)dx+p(1-\beta)\int_{-\infty}^{\bar{x}_p}g(x)dx}\mu_0+\notag\\ 
	&\frac{\beta\int_{\underline{x}_p}^{\bar{x}_p}g(x)dx}{(1-\beta)(1-p)+\beta\int_{\underline{x}_p}^{\bar{x}_p}g(x)dx+p(1-\beta)\int_{-\infty}^{\bar{x}_p}g(x)dx}E[x |x\in [\underline{x}_p,\bar{x}_p]]+\notag\\ 
	&\frac{p(1-\beta)\int_{-\infty}^{\bar{x}_p}g(x)dx}{(1-\beta)(1-p)+\beta\int_{\underline{x}_p}^{\bar{x}_p}g(x)dx+p(1-\beta)\int_{-\infty}^{\bar{x}_p}g(x)dx}E[x |x\in (-\infty,\bar{x}_p)]. \label{muNDextension1}	
\end{align}
We can simplify \eqref{xpupper}, \eqref{xplower}, and \eqref{muNDextension1} to the following two conditions:
\begin{align}
	&\frac{\beta}{(1-\beta)(1-p)}\int_{\underline{x}_p}^{\bar{x}_p}(G(x)-G(\underline{x}_p))dx+\frac{p}{1-p}\int_{-\infty}^{\bar{x}_p}G(x)dx-(\mu_0-\bar{x}_p)=0, \label{shortindiff_ex1}\\
	&\underline{x}_p+\bar{x}_p=\frac{2r_0\mu_0}{r}. \label{upper_xp1}
\end{align}
Plugging \eqref{upper_xp1} into \eqref{shortindiff_ex1} leads to 
\begin{align*}
	T(\bar{x}_p)\equiv &\frac{\beta}{(1-\beta)(1-p)}\int_{\frac{2r_0\mu_0}{r}-\bar{x}_p}^{\bar{x}_p}(G(x)-G(\frac{2r_0\mu_0}{r}-\bar{x}_p))dx\\&+\frac{p}{1-p}\int_{-\infty}^{\bar{x}_p}G(x)dx-(\mu_0-\bar{x}_p)=0.
\end{align*}
$T(\bar{x}_p)$ is obviously an increasing function of $\bar{x}_p$ as all three terms increase in $\bar{x}_p$. We also have $T(+\infty)>0$. However, the sign of $T(\frac{r_0\mu_0}{r})=\frac{p}{1-p}\int_{-\infty}^{\frac{r_0\mu_0}{r}}G(x)dx-(\mu_0-\frac{r_0\mu_0}{r})$ is indeterminate. There exists a unique $\bar{x}_p\in(\frac{r_0\mu_0}{r},+\infty)$ that solves $T(\bar{x}_p)=0$ if and only if $T(\frac{r_0\mu_0}{r})=\frac{p}{1-p}\int_{-\infty}^{\frac{r_0\mu_0}{r}}G(x)dx-(\mu_0-\frac{r_0\mu_0}{r})<0$. When $T(\frac{r_0\mu_0}{r})>0$, there does not exist an equilibrium in which $r\mu_{ND}>r_0\mu_0$, as there is no solution to $T(\bar{x}_p)=0$. 

Following a similar procedure, we can show that there exists an equilibrium where the investor takes a short position when they issue a simple report (that is, $r\mu_{ND}<r_0\mu_0$) if and only if $T(\frac{r_0\mu_0}{r})=\frac{p}{1-p}\int_{-\infty}^{\frac{r_0\mu_0}{r}}G(x)dx-(\mu_0-\frac{r_0\mu_0}{r})>0$. In this case, the thresholds are uniquely determined by equations \eqref{shortindiff_ex} and \eqref{upper_xp}.

To sum up, when $r>r_0$, we have $r\mu_{ND}<r_0\mu_0$ if and only if the following condition holds:
\begin{align}
V(r)\equiv\frac{p}{1-p}\int_{-\infty}^{\frac{r_0\mu_0}{r}}G(x)dx-(\mu_0-\frac{r_0\mu_0}{r})>0.
\end{align}
Note that $V(r)$ is a decreasing function of $r$ and $V(r_0)>0$. Thus, when $V(+\infty)=\frac{p}{1-p}\int_{-\infty}^{0}G(x)dx-\mu_0<0$, $r\mu_{ND}<r_0\mu_0$ if and only if $r<\bar{r}$, where $\bar{r}$ is uniquely determined by
\begin{align}\label{barr}
\frac{p}{1-p}\int_{-\infty}^{\frac{r_0\mu_0}{\bar{r}}}G(x)dx-(\mu_0-\frac{r_0\mu_0}{\bar{r}})=0.
\end{align}

When $V(+\infty)=\frac{p}{1-p}\int_{-\infty}^{0}G(x)dx-\mu_0>0$, we have $V(r)>0$ for any $r$, and thus we are at a corner case where $\bar{r}=+\infty$. In this case, $r\mu_{ND}<r_0\mu_0$ always holds.

\subsubsection*{Step 3: Prove that the investor always disclosing $r$ is an equilibrium}
In the strategy profile defined in Parts 2-3 of Proposition \ref{prop:firm response}, it is obvious that all types earn positive expected profit. Upon observing the message $m=(\emptyset,\emptyset)$, the market believes that the investor's type is $t=(\emptyset,\emptyset)$ and sets price $P_2(\emptyset,\emptyset)=r_0\mu_0$. Thus, if the investor withholds $r$, their profit is $0$. This means no type has incentive to deviate from the strategies in Parts 2-3 of Proposition \ref{prop:firm response}. 

\subsection*{Proof of Corollary \ref{cs3} for the case $r<\bar{r}$}
Let $K(\underline{x}_p,\bar{x}_p,p,\beta)\equiv\frac{\beta}{1-\beta}\int_{\underline{x}_p}^{\bar{x}_p}(G(x)-G(\bar{x}_p))dx+\frac{p}{1-p}\int_{-\infty}^{\underline{x}_p}G(x)dx-(\mu_0-\underline{x}_p)$.

Differentiating \eqref{shortindiff_ex} and \eqref{upper_xp} with respect to $\beta$ and applying the Implicit Function Theorem, we have
 \begin{align*}
 	&\frac{\partial \bar{x}_p}{\partial \beta}=\frac{(1-p)\frac{K_{\beta}}{K_{\underline{x}_p}}}{(1+p)-(1-p)\frac{K_{\bar{x}_p}}{K_{\underline{x}_p}}}, \\
 	&\frac{\partial \underline{x}_p}{\partial \beta}=\frac{(1+p)\frac{K_{\beta}}{K_{\bar{x}_p}}}{(1-p)-(1+p)\frac{K_{\underline{x}_p}}{K_{\bar{x}_p}}}.
 \end{align*}
 Since $K_{\beta}<0$, $K_{\underline{x}_p}>0$ and $K_{\bar{x}_p}<0$, we have $\frac{\partial \bar{x}_p}{\partial \beta}<0$ and $\frac{\partial \underline{x}_p}{\partial \beta}>0$.

\subsection*{Proof of Corollary \ref{cs3} for the case $r>\bar{r}$}
Let $S(\underline{x}_p,\bar{x}_p,p,\beta)\equiv\frac{\beta}{(1-\beta)(1-p)}\int_{\underline{x}_p}^{\bar{x}_p}(G(x)-G(\underline{x}_p))dx+\frac{p}{1-p}\int_{-\infty}^{\bar{x}_p}G(x)dx-(\mu_0-\bar{x}_p)$.

Differentiating \eqref{shortindiff_ex1} and \eqref{upper_xp1} with respect to $\beta$ and applying the Implicit Function Theorem, we have
\begin{align*}
	&\frac{\partial \bar{x}_p}{\partial \beta}=\frac{\frac{S_{\beta}}{S_{\underline{x}_p}}}{1-\frac{S_{\bar{x}_p}}{S_{\underline{x}_p}}}, \\
	&\frac{\partial \underline{x}_p}{\partial \beta}=\frac{\frac{S_{\beta}}{S_{\bar{x}_p}}}{1-\frac{S_{\underline{x}_p}}{S_{\bar{x}_p}}}.
\end{align*}
 Since $S_{\beta}>0$, $S_{\underline{x}_p}<0$ and $S_{\bar{x}_p}>0$, we have $\frac{\partial \bar{x}_p}{\partial \beta}<0$ and $\frac{\partial \underline{x}_p}{\partial \beta}>0$.
\newpage
\singlespacing
\bibliographystyle{apalike}
\bibliography{short1}
\end{document}